\documentclass[10pt]{amsart} 
\usepackage{enumerate}
\usepackage[all]{xy}
\usepackage{graphicx}
\usepackage{footnote}
\usepackage{hyperref}
\usepackage{amssymb, latexsym, amscd, amsmath,stmaryrd,wrapfig,appendix}
\usepackage{wasysym}
\usepackage{proof}

\theoremstyle{plain}
\newtheorem{lemma}{Lemma}[section]
\newtheorem{theorem}[lemma]{Theorem}
\newtheorem{corollary}[lemma]{Corollary}

\newtheorem{proposition}[lemma]{Proposition}

\theoremstyle{definition}

\newtheorem{definition}[lemma]{Definition}
\newtheorem{example}[lemma]{Example}
\newtheorem{remark}[lemma]{Remark}
\numberwithin{equation}{section}

\begin{document}

%
%
%
%

\title{Lax functors and coalgebraic weak bisimulation}
\author[T. Brengos]{Tomasz Brengos}
\email{t.brengos@mini.pw.edu.pl}
\thanks{}

\address{Faculty of Mathematics and Information Science\\
         Warsaw University of Technology\\ 
         ul. Koszykowa 75 \\       
         00-662 Warszawa, Poland}

\maketitle
\begin{abstract}
We generalize the work by Soboci\'nski on relational presheaves and their connection with weak (bi)simulation for labelled transistion systems to a coalgebraic setting. We show that the coalgebraic notion of saturation studied in our previous work can be expressed in the language of lax functors in terms of existence of a certain adjunction between categories of lax functors. This observation allows us to generalize the notion of the coalgebraic weak bisimulation to lax functors. We instantiate this definition on two examples of timed systems and show that it coincides with their time-abstract behavioural equivalence. 
\end{abstract}

\tableofcontents

\section{Introduction}

We have witnessed a rapid development of the theory of coalgebras as a unifying theory for state-based systems \cite{gumm:elements,rutten:universal}.  A coalgebra can be thought of as an abstract representation of a single step of computation of a given process. The theory of coalgebras provides a good setting for the study of bisimulation \cite{rutten:universal,staton11}. The notion of a strong bisimulation for different transition systems plays an important role in theoretical computer science. Weak bisimulation is a relaxation of this notion by allowing the so-called $\tau$-moves, i.e. silent, unobservable transitions.  One of several (equivalent) ways to define Milner's weak bisimulation \cite{milner:cc,sangiorgi2011:bis} on a labelled transition system $\alpha$ is to consider it as a strong bisimulation on its closure $\alpha^{*}$.  Labelled transition systems closure is reduced to finding the smallest LTS containing all transitions of the original structure and satisfying the rules \cite{sangiorgi2011:bis}: 

\begin{align}
\infer{x\stackrel{\tau}{\to} x}{}\qquad \infer{x\stackrel{a}{\to} x''}{x\stackrel{a}{\to} x' & x'\stackrel{\tau}{\to} x''} \qquad \infer{x\stackrel{a}{\to} x''}{x\stackrel{\tau}{\to} x' & x'\stackrel{a}{\to} x''} \label{rules_saturated_coalgebra}
\end{align}

It has been shown in \cite{brengos2014:cmcs,brengos2015:lmcs,brengos2015:jlamp} that from the point of view of the theory of coalgebra the systems with silent moves should be considered as coalgebras over a monadic type. This allows us to abstract away from a specific structure on labels and consider systems of the type $X\to TX$ for a monad $T$. A coalgebra $X\to TX$ is an endomorphism in the so-called Kleisli category for the monad $T$. The rules presented above that describe a closed LTS structure can be restated in terms of the composition in an order enriched Kleisli category by the following two axioms~\cite{brengos2015:lmcs}:
\begin{align}
id\leq \alpha \text{ and } \alpha \circ \alpha \leq \alpha. \label{ineq:monads_in}
\end{align}
Intuitively, these two rules say that a coalgebra satisfying them is reflexive and transitive. 
Hence, coalgebraically, Milner's weak bisimulation on a labelled transition system $\alpha$ is, in fact, a strong bisimulation on the reflexive and transitive closure of $\alpha$ taken in the suitable Kleisli category.\footnote{Categorically, any endomorphism from an order enriched category which satisfies (\ref{ineq:monads_in}) is a monad \emph{in} the underlying category. Thus, labelled transition systems saturation may be also viewed as assigning to a given LTS the free monad generated by it \cite{brengos2015:lmcs,brengos2015:jlamp}.} However, there are several examples of systems with silent moves (e.g. fully probabilistic systems \cite{baier97:cav}) for which the definition of weak bisimulation is more subtle and is not a strong bisimulation on the reflexive and transitive closure of a system. 

In \cite{brengos2015:jlamp} we presented a general coalgebraic setting where we introduce the notion of weak bisimulation. It encompasses an extensive list of transition systems among which we find labelled transition systems, Segala systems \cite{segalalynch94:concur,brengos2015:lmcs} and fully probabilistic systems. Moreover, we identified the condition under which weak bisimulation can be defined as a strong bisimulation on a reflexive and transitive closure. The condition is \emph{saturation admittance} of the underlying category. Saturation can be intuitively understood as a reflexive and transitive closure which additionally preserves weak homomorphisms. Although all examples of coalgebras considered in \cite{brengos2015:jlamp} admit reflexive and transitive closure, not all admit {saturation}.  Saturation  is the key ingredient in the definition of weak bisimulation. Therefore, in order to define weak bisimulation one has to consider saturation admittance first. In case of its absence, we move from the underlying category to a new category with the desired properties and perform saturation there. 
 
More formally, a given category admits (coalgebraic) saturation if there is a left adjoint to the inclusion functor from the category of reflexive and transitive endomorphisms to the category of all endomorphisms \cite{brengos2015:jlamp}. The link between (not necessarily coalgebraic) weak bisimulation saturation and existence of a certain adjunction has emerged several times in the literature before. To our knowledge, this goes back to \cite{fiore99} where the authors present the definition of weak bisimilarity for presheaves. The main component of their construction is an adjunction between certain slice categories. Nevertheless, the motivations for our paper stem mainly from more recent work on relational presheaves, i.e. lax functors $\mathbb{D}\to \mathsf{Rel}$ \cite{sobocinski:jcss}. Soboci\'nski in \emph{loc. cit.} shows that different types of systems, including labelled transition systems, may be viewed as relational presheaves and that  labelled transition system saturation may be encoded in terms of an adjunction between certain categories of relational presheaves, where the right adjoint is the so-called change-of-base functor. The category $\mathsf{Rel}$ of sets as objects and binary relations as morphisms is isomorphic to the Kleisli category for the powerset monad $\mathcal{P}$. If we generalize Soboci\'nski's approach to lax functors whose codomain is the Keisli category for an arbitrary monad $T$ we obtain results which are consistent with our previous work on coalgebraic saturation and weak bisimulation \cite{brengos2015:lmcs,brengos2015:jlamp}. 
 
\subsection*{Content and organization of the paper} The main contributions of the paper are the following:
\begin{itemize}
\item we generalize the results by Soboci\'nski on existence of a left adjoint to the change-of-base functor to the setting of arbitrary lax functors whose codomain is an order enriched category;
\item  we show that coalgebraic saturation and weak bisimulation from our previous work \cite{brengos2015:lmcs,brengos2015:jlamp} and Soboci\'nski's work on relational presheaves have a common denominator. Coalgebraic saturation can, in fact, be understood as a consequence of existence of the left adjoint to the change-of-base functor between certain categories of lax functors;
\item we define weak bisimulation for lax functors. This relation takes into account a cumulative behaviour of a lax functor. We show that there are examples of timed systems (e.g. timed processes semantics \cite{Larsen199775} or continuous time Markov chains \cite{books/daglib/0095301}) which can be naturally modelled as lax functors, for which weak bisimulation turns out to be the so-called time-abstract bisimulation.
\end{itemize}

The paper is organized as follows. In Section \ref{section:basics} we recall basic definitions and properties required in the remainder of the paper. In Section \ref{section:lax_functors} we consider the notion of a lax functor and present several examples of lax functor categories. We relate some of them to certain categories of coalgebras. Here, we also focus on the change-of-base functor and existence of its left adjoint. Section \ref{section:weak_bisimulation_lax_functors} is devoted to the presentation of the
notion of weak bisimulation for lax functors. We instantiate this definition on timed processes semantics and continuous time Markov chains and show that it models their time-abstract behavioural equivalence. Moreover, we observe that the lax functorial weak bisimulation extends the coalgebraic weak bisimulation. 

\section{Basic notions}\label{section:basics}
We assume the reader is familar with the following basic category theory notions: a category, a functor, a monad and an adjunction (see e.g. \cite{maclane:cwm}  for an introduction to category theory). We will now briefly recall some of them here and also present other basics needed in this paper.

 For a family  $\{X_i\}_{i\in I}$ of objects and a family $\{f_i:X_i\to Y\}_{i\in I}$ of morphisms in a given category if the coproduct of  $\{X_i\}$ exists then we denote it by $\sum_{i}X_i$ and denote the cotuple from $\sum_i X_i$ to $Y$ by $[\{f_i\}]:\sum_i X_i \to Y$ or simply by $[f_i]$. In this case, $\mathsf{in}_i:X_i\to \sum_i X_i$ is the coprojection into the $i$-th component of $\sum_i X_i$.

\subsection{Coalgebras}
Let $\mathsf{C}$ be a category and $F \colon \mathsf{C}\rightarrow \mathsf{C}$ a functor.  An \emph{$F$-coalgebra} is a morphism $\alpha:X\to FX$ in $\mathsf{C}$.  The domain $X$ of $\alpha$ is called \emph{carrier} and the morphism  $\alpha$ is sometimes also called \emph{structure}. A \emph{homomorphism} from an $F$-coalgebra $\alpha:X\to FX$ to an $F$-coalgebra $\beta:Y\to FY$  is an arrow $f \colon X\rightarrow Y$ in $\mathsf{C}$ such that $ F(f)\circ \alpha =  \beta \circ f$. The category of all $F$-coalgebras and homomorphisms between them is denoted by $\mathsf{C}_F$. Many transition systems can be captured by the notion of coalgebra. The most important from our perspective are listed below. 

Let $\Sigma$ be a fixed set and put $\Sigma_\tau = \Sigma+\{\tau\}$. The label $\tau$ is considered a special label called \emph{silent} or \emph{invisible} label. 

\subsubsection{Labelled transition systems} $\mathcal{P}(\Sigma_\tau\times \mathcal{I}d)$-coalgebras are  labelled transition systems over the alphabet $\Sigma_\tau$ \cite{milner:cc,rutten:universal, sangiorgi2011:bis}. Here, $\mathcal{P}$ denotes the powerset functor. In this paper we also consider labelled transition systems with a monoid structure on labels, i.e. coalgebras of the type $\mathcal{P}(M\times \mathcal{I}d)$,  or even more generally, as coalgebras of the type $\mathcal{P}(\Sigma_\tau\times M\times \mathcal{I}d)$ for a monoid $(M,\cdot,1)$. 


\subsubsection{Fully probabilistic systems} Originally, fully probabilistic systems \cite{baier97:cav} were modelled as $\mathcal{D}(\Sigma_\tau\times \mathcal{I}d)$-coalgebras \cite{sokolova11}, where $\mathcal{D}$ denotes the distribution functor. However, following the guidelines of \cite{mp2013:weak-arxiv,gp:icalp2014,brengos2015:jlamp}, in this paper we extend the type of these systems and  consider them as $\mathbb{F}_{[0,\infty]}(\Sigma_\tau\times \mathcal{I}d)$-coalgebras. Here, $\mathbb{F}_{[0,\infty]}$ is the $\mathsf{Set}$-endofunctor defined for any  set $X$ and map $f:X\to Y$ by:
\begin{align*}
&\mathbb{F}_{[0,\infty]}X = \{ \phi:X\to [0,\infty] \mid \text{supp } \phi \text{ is at most countable}\},\\
&\mathbb{F}_{[0,\infty]}f(\phi)(y) = \sum_{x:y=f(x)}\phi(x) \text{ for }\phi \in \mathbb{F}_{[0,\infty]}X,
\end{align*}
where $[0,\infty]$ is the semiring $([0,\infty],+,\cdot)$ of non-negative real numbers with infinity with standard addition and multiplication. Note that the distribution functor $\mathcal{D}$ is a subfunctor of $\mathbb{F}_{[0,\infty]}$ and, hence, any coalgebra $X\to \mathcal{D}(\Sigma_\tau\times X)$ can be naturally translated into $X\to \mathbb{F}_{[0,\infty]}(\Sigma_\tau\times X)$ \cite{brengos2015:jlamp}. 

In order to simplify our notation we will sometimes denote $\phi \in \mathbb{F}_{[0,\infty]}X$ by $\sum_{x\in X}\phi(x)\cdot x$ or $\sum_{i=1}^\infty \phi(x_i)\cdot x_i$ if $\text{supp }\phi =\{x_1,x_2,\ldots\}$.

\subsubsection{Filter coalgebras} These are systems of the type $\mathcal{F}$, where $\mathcal{F}:\mathsf{Set}\to\mathsf{Set}$ denotes the filter functor which assigns to any set $X$ the set $\mathcal{F}X$ consisting of all filters on $X$ and to $f:X\to Y$ the map
$
\mathcal{F}f:\mathcal{F}X\to \mathcal{F}Y$ assigning to any filter $\mathcal{G}$ the filter $\{Y'\subseteq Y\mid f^{-1}(Y')\in \mathcal{G}\}$.
See e.g. \cite{gumm:elements} for details.
\subsection{Monads and their Kleisli categories} 
A \emph{monad} on $\mathsf{C}$ is a triple $(T,\mu,\eta)$, where $T:\mathsf{C}\to \mathsf{C}$ is an endofunctor and $\mu:T^2\implies T$, $\eta:\mathcal{I}d\implies T$ are two natural transformations for which the following diagrams commute:
$$
\xymatrix{
T^3\ar[d]_{T\mu} \ar[r]^{\mu_T} & T^2\ar[d]^{\mu} & & T\ar[d]_{\eta_T}\ar[r]^{T\eta}\ar@{=}[dr] & T^2\ar[d]^{\mu} \\
T^2 \ar[r]_{\mu} & T & &T^2 \ar[r]_\mu & T
}
$$
The transformation $\mu$ is called  \emph{multiplication} and $\eta$ \emph{unit}.
Any monad gives rise to the Kleisli category for $T$. To be more precise, if $(T,\mu,\eta)$ is a monad on a category $\mathsf{C}$ then the \emph{Klesli category} $\mathcal{K}l(T)$ for $T$ has the class of objects equal to the class of objects of $\mathsf{C}$ and for two objects $X,Y$ in $\mathcal{K}l(T)$ we have ${\mathcal{K}l(T)}(X,Y) = {\mathsf{C}}(X,TY)$
with the composition in $\mathcal{K}l(T)$ defined between two morphisms $f:X\to TY$ and $g:Y\to TZ$ by $\mu_Z \circ T(g) \circ f$. The category $\mathsf{C}$ is a subcategory of $\mathcal{K}l(T)$ where the inclusion functor $^{\sharp}$ sends each object $X\in \mathsf{C}$ to itself and each morphism $f:X\to Y$ in $\mathsf{C}$ to the morphism $f^{\sharp}:X\to TY; f^{\sharp} = \eta_Y\circ f$. 

\subsubsection{Powerset monad}
The powerset endofunctor $\mathcal{P}:\mathsf{Set}\to \mathsf{Set}$ is a monad whose multiplication $\mu:\mathcal{P}^2\implies\mathcal{P}$ and unit $\eta:\mathcal{I}d\implies \mathcal{P}$ are given on their $X$-components by: $\mu_X:\mathcal{P}\mathcal{P}X\to \mathcal{P}X; S \mapsto \bigcup S \text{ and }\eta_X:X\to \mathcal{P}X; x\mapsto \{x\}$.
The category $\mathcal{K}l(\mathcal{P})$ consists of sets as objects and maps of the form $X\to \mathcal{P}Y$ as morphisms. For $f:X\to \mathcal{P}Y$ and $g:Y\to \mathcal{P}Z$ the composition $g\circ f:X\to \mathcal{P}Z$ is as follows:
\begin{align*}
& g\circ  f(x) = \bigcup g(f(x))= \{z \mid z\in g(y) \text{ \& }y\in f(x) \text{ for some }y\in Y\}. 
\end{align*} 
For any two sets $X,Y$ there is a bijective correspondence between maps $X\to \mathcal{P}Y$ and binary relations between elements of $X$ and $Y$. 
Indeed, for $f:X\to \mathcal{P}Y$ we put $R_f\subseteq X\times Y$, $(x,y)\in R_f \iff y\in f(x)$ and for $R\subseteq X\times Y$ we define $f_R:X\to \mathcal{P}Y;x\mapsto \{y\mid x R y\}$. It is now easy to see that the category $\mathcal{K}l(\mathcal{P})$ is isomorphic to the category $\mathsf{Rel}$ of sets as objects, binary relations as morphisms and relation composition as morphism composition.

\subsubsection{LTS monad(s)}\label{subsection:lts_monads}
Labelled transition systems functor $\mathcal{P}(\Sigma_\tau\times \mathcal{I}d)$ carries a monadic structure 
$(\mathcal{P}(\Sigma_\tau\times \mathcal{I}d),\mu,\eta)$ \cite{brengos2015:lmcs}, where the $X$-components of $\eta$ and $\mu$ are:
\begin{align*}
&\eta_X(x) = \{(\tau,x)\} \text{ and } \mu_X(S) = \hspace{-0.3cm} \bigcup_{(\sigma,S')\in S}\hspace{-0.3cm} \{(\sigma,x)\mid (\tau,x)\in S'\}\cup \hspace{-0.3cm} \bigcup_{(\tau,S')\in S} \hspace{-0.3cm} S'.
\end{align*}
For $f:X\to \mathcal{P}(\Sigma_\tau\times Y)$ and $g:Y\to \mathcal{P}(\Sigma_\tau\times Z)$ the composition $g\circ f$ in $\mathcal{K}l(\mathcal{P}(\Sigma_\tau\times \mathcal{I}d))$ is
$
g\circ f(x)= \{(\sigma,z)\mid x\stackrel{\sigma}{\to}_f y \stackrel{\tau}{\to}_g z \text{ or } x\stackrel{\tau}{\to}_f y \stackrel{\sigma}{\to}_g z\},
$
where $x\stackrel{\sigma}{\to}_f y$ denotes $(\sigma,y)\in f(x)$.

As mentioned above, we will sometimes consider labelled transition systems as coalgebras of the type $\mathcal{P}(M\times \mathcal{I}d)$, or even more generally as coalgebras of the type $\mathcal{P}(\Sigma_\tau\times M\times \mathcal{I}d)$ for a monoid $M=(M,\cdot,1)$.  The latter functor carries a monadic structure which is a consequence of an application of the writer monad transformer to the LTS monad defined above with $M$ as the argument \cite{liang95}. Its unit and multiplication are given on their components as follows:
\begin{align*}
&\eta_X(x)= \{(\tau,1,x)\},\\
 & \mu_X(S)= \hspace{-0.5cm}\bigcup_{(\sigma,m,S')\in S} \hspace{-0.5cm} \{(\sigma,m\cdot n,x)\mid (\tau,n,x)\in S'\}\cup \hspace{-0.5cm} \bigcup_{(\tau,m,S')\in S}\hspace{-0.5cm} \{(\sigma,m\cdot n, x) \mid (\sigma,n,x)\in S'\}. 
\end{align*}
The composition $g\circ f$  of $f:X\to \mathcal{P}(\Sigma_\tau \times M\times Y)$ and $g:Y\to \mathcal{P}(\Sigma_\tau \times M\times Z)$ in the Kleisli category for this monad is:
$$
g\circ f(x) =  \{(\sigma,m\cdot n,z) \mid  x\stackrel{(\sigma,m)}{\to}_f y \stackrel{(\tau,n)}{\to}_g z \text{ or } x\stackrel{(\tau,m)}{\to}_f y \stackrel{(\sigma,n)}{\to}_g z \}.
$$
We see that if $\Sigma =\varnothing$ then $\mathcal{P}(\Sigma_\tau \times M\times \mathcal{I}d) \cong \mathcal{P}(M\times \mathcal{I}d)$ and if $M=1$ is the one-element monoid then $\mathcal{P}(\Sigma_\tau \times M\times \mathcal{I}d)\cong \mathcal{P}( \Sigma_\tau\times \mathcal{I}d)$. Whenever $\Sigma=\varnothing$ and $M=1$ then this monad becomes (isomorphic to) $\mathcal{P}$. 

In order to simplify our notation we will often denote $\mathcal{P}(\Sigma_\tau\times M\times \mathcal{I}d)$ by $\mathcal{P}^{\Sigma,M}$.

\subsubsection{Quantale valued monad} 
Let $\mathcal{Q}=(\mathcal{Q},\cdot ,1,\leq)$ be a \emph{unital quantale}, i.e. a relational structure for which 
\begin{enumerate}
\item $(\mathcal{Q},\cdot,1)$ is a monoid,
\item $(\mathcal{Q},\leq)$ is a complete lattice,
\item arbitrary suprema are preserved by the monoid multiplication.
\end{enumerate}
An arbitrary unital quantale $\mathcal{Q}$ gives rise to the $\mathsf{Set}$-based monad $\mathcal{Q}^{(-)}$, called \emph{quantale valued monad}, which assigns to any set $X$ the set of all functions $\mathcal{Q}^X$ from $X$ to $\mathcal{Q}$ and to any map $f:X\to Y$ the map $\mathcal{Q}^f:\mathcal{Q}^X\to \mathcal{Q}^Y$ given by:
\begin{align*}
\mathcal{Q}^f(\phi)(y) = \bigvee_{x:f(x)=y} \phi(x).
\end{align*}
The monadic structure $(\mathcal{Q}^{(-)},\mu,\eta)$ is given by the $X$-components of $\eta$ and $\mu$:
\begin{align*}
\eta_X(x)(x') = \left\{ \begin{array}{cc} 1 & \text{ if } x=x' \\ \perp & \text{ otherwise.} \end{array} \right . \text{ and }\mu_X(\phi)(x) = \bigvee_{\psi\in \mathcal{Q}^X} \phi(\psi)\cdot \psi(x).
\end{align*}
The Kleisli category $\mathcal{K}l(\mathcal{Q}^{(-)})$ is isomorphic to the category $\mathsf{Mat}(\mathcal{Q})$ of sets as objects and $\mathcal{Q}$-matrices, i.e. functions $X\times Y\to \mathcal{Q}$, as morphisms between $X$ and $Y$ \cite{rosen:quanta96}.

Interestingly, the LTS monad $\mathcal{P}(\Sigma_\tau \times M\times \mathcal{I}d)$ for a monoid $M=(M,\cdot,1)$ can be viewed as an example of a quantale valued monad. Indeed, it is isomorphic to a quantale valued monad for the quantale $(\mathcal{P}({\Sigma_\tau\times M}),\cdot,\{(\tau,1)\},\subseteq)$, where $\cdot$ is given by:
$$
A\cdot B = \{ (\sigma,m\cdot n) \mid (\tau,m)\in A \text{ and }(\sigma,n)\in B \text{ or }(\sigma,m)\in A \text{ and }(\tau,n)\in B\}.
$$ 
\subsubsection{Fully probabilistic systems monad} Before we elaborate more on the fully probabilistic systems functor we focus on describing the monadic structure of $\mathbb{F}_{[0,\infty]}$ first. Multiplication $\mu$ and unit $\eta$ of this monad are given on their $X$-components by \cite{brengos2015:jlamp,gp:icalp2014,mp2013:weak-arxiv}:
\begin{align*}
\eta_X(x) = 1\cdot x \text{ and }
\mu_X(\phi) = \sum_{\psi \in \mathbb{F}_{[0,\infty]}X} \phi(\psi) \cdot \psi.
\end{align*}

The fully probabilistic systems functor $\mathbb{F}_{[0,\infty]}(\Sigma_\tau\times \mathcal{I}d)$ is a monad which arises as a consequence of a general construction of a monadic structure on the functor $T(\Sigma_\tau \times \mathcal{I}d)$ for a $\mathsf{Set}$-based monad $T$ presented in \cite{brengos2015:lmcs} and applied to the fully probabilistic systems functor case in \cite{brengos2015:jlamp}. The composition in $\mathcal{K}l(\mathbb{F}_{[0,\infty]}(\Sigma_\tau\times \mathcal{I}d))$ is given for $f:X\to \mathbb{F}_{[0,\infty]}(\Sigma_\tau\times Y)$ and $g:Y\to \mathbb{F}_{[0,\infty]}(\Sigma_\tau\times Z)$ by \cite{brengos2015:jlamp}:
\begin{align*}
&g\circ f(x)(\sigma,z) =
\left \{ \begin{array}{cc} \sum_{y\in Y}g(y)(\tau,z)\cdot f(x)(\tau,y) & \text{if }\sigma = \tau,\\ \sum_{y\in Y}g(y)(\sigma,z)\cdot f(x)(\tau,y) + g(y)(\tau,z)\cdot f(x)(\sigma,y) & \text{otherwise.} \end{array} \right. 
\end{align*}

\subsubsection{Filter monad}
 The filter functor $\mathcal{F}$ carries a monadic structure $(\mathcal{F},\mu,\eta)$ given by (see e.g. \cite{gahler91}):
\begin{align*}
&\mu_X:\mathcal{F}\mathcal{F}X\to \mathcal{F}X; \mathcal{G}\mapsto \mu_X(\mathcal{G}),\text{ and }\eta_X:X\to \mathcal{F}X; x\mapsto \{U\subseteq X\mid x\in U\},
\end{align*}
where $\mu_X(\mathcal{G}) = \{A\subseteq X\mid A^\mathcal{F} \in \mathcal{G}\}$ with $A^\mathcal{F} = \{\mathcal{H}\in\mathcal{F}X\mid A\in \mathcal{H} \}$ defined as the set of all filters on $X$ containing $A$.

\subsection{Coalgebras with internal moves}\label{subsection:coalgebras_with_internal} Originally \cite{hasuojacobssokolova2006:jssst,silvawesterbaan2013:calco}, coalgebras with internal moves were introduced in the context of coalgebraic trace semantics as coalgebras of the type $T(F+\mathcal{I}d)$ for a monad $T$ and an endofunctor $F$ on a common category. In \cite{brengos2015:lmcs} we showed that given some mild assumptions on $T$ and $F$ we may either introduce a monadic structure on $T(F+\mathcal{I}d)$ or embed it into the monad $TF^{*}$, where $F^{*}$ is the free monad over $F$. It is worth noting here that, both, the LTS monad $\mathcal{P}(\Sigma_\tau\times \mathcal{I}d)$ and the fully probabilistic systems monad $\mathbb{F}_{[0,\infty]}(\Sigma_\tau\times \mathcal{I}d)$ arise by the application of the first construction \cite{brengos2015:lmcs,brengos2015:jlamp}.

The trick of modelling  the silent steps via a monad allows us \emph{not} to specify the internal moves explicitly. Instead of considering $T(F+\mathcal{I}d)$-coalgebras we consider $T'$-coalgebras for a monad $T'$. Hence, the term ``coalgebras with internal moves" becomes synonymous to ``coalgebras whose type is a monad".  Since coalgebras with silent transitions are of primary interest to this paper, we assume, unless stated otherwise, that all coalgebras considered here are systems whose type functor $T$ carries a monadic structure. 

To give a $T$-coalgebra is to give an endomorphism in $\mathcal{K}l(T)$. We use this observation and present our results in as general setting as possible. Hence, we will often replace $\mathcal{K}l(T)$ with an arbitrary  category $\mathsf{K}$ and work in the context of endomorphisms of $\mathsf{K}$ bearing in mind our prototypical example of $\mathsf{K}=\mathcal{K}l(T)$. 
%

\subsection{Order enriched categories}
In this paper we work with standard $2$-category and enriched category notions. However, since the only type of enrichment we consider is (several types of) order enrichment, we recall basic definitions only from the point of view of the structures we are interested in. The reader is referred to e.g. \cite{kelly82,lack:09} for a more general perspective.
 A category is \emph{order enriched} if it is $\mathsf{Pos}$-enriched, where $\mathsf{Pos}$ denotes the category of all posets and monotonic maps between them. In other words, a category is order enriched if each hom-set is a poset with the order preserved by the composition.  A  functor between two order enriched categories is \emph{locally monotonic} if it preserves the order. 
  
    We will be often interested in stronger types of enrichment: $\mathsf{Sup}$, $\omega\mathsf{Cpo}$-, $\mathsf{DCpo}$-, $\omega\mathsf{Cpo}^\vee$- and $\mathsf{DCpo}^\vee$-enrichment. Here, $\mathsf{Sup}$ denotes the category of posets which admit arbitrary suprema as objects and maps that preserve suprema as morphisms. The category $\omega\mathsf{Cpo}$ consists of partially ordered sets that admit suprema of ascending $\omega$-chains as objects and maps that preserve them as morphisms. $\mathsf{DCpo}$ is the category of posets that admit suprema of arbitrary directed sets and Scott-continuous maps.  Any $\mathsf{DCpo}$-enriched category  is also $\omega\mathsf{Cpo}$-enriched. Finally, $\omega\mathsf{Cpo}^\vee$ and $\mathsf{DCpo}^\vee$ are full subcategories of $\omega\mathsf{Cpo}$ and $\mathsf{DCpo}$ respectively whose objects admit binary joins. Note that all of these five categories are enriched over themselves.
    
    The order enrichment in a $\omega\mathsf{Cpo}^\vee$- or a $\mathsf{DCpo}^\vee$-enriched category only guarantees that for any morphisms with suitable domain and codomain we have:
$$
	f\circ h \vee g\circ h \leq (f\vee g) \circ h  
	\text{ and } h\circ f \vee  h\circ g \leq h\circ (f\vee g). 
$$ 
If the first (the second) inequality becomes an equality then we say that the given category is \emph{right} (resp. \emph{left}) \emph{distributive}. 


  A functor-like assignment $\pi$ from a category $\mathbb{D}$ to an order enriched category $\mathsf{K}$  is called \emph{lax functor} if:
\begin{itemize}
\item $id_{\pi D}\leq \pi(id_D)$ for any object $D\in \mathbb{D}$,
\item $\pi(d_1)\circ \pi(d_2) \leq \pi(d_1\circ d_2)$ for any two composable morphisms $d_1,d_2\in \mathbb{D}$.
\end{itemize} 
Let $\pi,\pi':\mathbb{D}\to \mathsf{K}$ be two lax functors. A family $f=\{f_D:\pi D\to\pi'D\}_{D\in \mathbb{D}}$ of morphisms in $\mathsf{K}$ is called \emph{lax natural transformation} if  for any $d:D\to D'$ in $\mathbb{D}$ we have
$f_{D'}\circ \pi(d)\geq \pi'(d) \circ f_D$. \emph{Oplax functors} and \emph{oplax transformations} are defined by reversing the order in the above. Note that in the more general 2-categorical setting an (op)lax functor and an (op)lax natural transformation are assumed to additionally satisfy extra coherence conditions \cite{lack:09}.  In our setting of order enriched categories these conditions are vacuously true, hence we do not list them here.

Let $\mathsf{K}$ and $\mathsf{K}'$ be two order enriched categories. Given two locally monotonic functors $F:\mathsf{K}\to \mathsf{K}'$ and $U:\mathsf{K}'\to \mathsf{K}$ a \emph{$2$-adjunction} is a family of isomorphisms of posets $\{\phi_{X,Y}:\mathsf{K}'(FX,Y) \cong \mathsf{K}(X,UY)\}_{X\in \mathsf{K}, Y\in \mathsf{K}'}$ natural in $X$ and $Y$. In this case $F$ and $U$ are called \emph{left-} and \emph{right 2-adjoint} respectively. Finally, a locally monotonic faithful functor $F:\mathsf{K}\to \mathsf{K}'$  is said to be \emph{locally reflective} provided that for any objects $X,Y\in \mathsf{K}$ the restriction $F_{X,Y}:\mathsf{K}(X,Y)\to \mathsf{K}'(FX,FY)$ is a functor between posets $\mathsf{K}(X,Y)$ and $\mathsf{K}'(FX,FY)$ viewed as categories which additionally admits a left adjoint. In this case the order enriched category $\mathsf{K}$ is called \emph{locally reflective subcategory} of $\mathsf{K}'$. 

The Kleisli category for monads considered in the previous subsection is order enriched with the order on hom-sets imposed by a natural pointwise order on $TY$, whose definition and properties are summarized in the table below. Since whenever $\Sigma=\varnothing$ we have $\mathbb{F}_{[0,\infty]}\cong \mathbb{F}_{[0,\infty]}(\Sigma_\tau~\times~\mathcal{I}d)$ the monad $\mathbb{F}_{[0,\infty]}$ is not mentioned below explicitly.  For $f,g:X\to TY$ in $\mathcal{K}l(T)$ for a suitable monad we have:
\begin{savenotes}
\begin{table}[h]
\begin{tabular}{|p{2.3cm}|p{4cm}||p{1cm}|p{0.7cm}|p{0.7cm}|p{1cm}|}
\hline 
Monads & $f\leq g$ if and only if &  $\mathsf{DCpo}^\vee$-enr.& left dist. & right dist. &Ref.  \\\hline \hline
$\mathcal{Q}^{(-)}$ & $\forall x\in X$, $\forall y\in Y$ $f(x)(y)~\leq~g(x)(y),$ &\checked &\checked  &\checked &  \footnote{Any quantale valued monad $\mathcal{Q}^{(-)}$ (hence, also the LTS monad $\mathcal{P}(\Sigma_\tau \times M\times \mathcal{I}d)$)  yields a $\mathsf{Sup}$-enriched  Kleisli category $\mathcal{K}l(\mathcal{Q}^{(-)})\cong \mathsf{Mat}(\mathcal{Q})$ \cite{rosen:quanta96}. } \\ \hline 
$\mathbb{F}_{[0,\infty]}(\Sigma_\tau\times~\mathcal{I}d)$ & $\forall x\in X,  y\in Y, \sigma\in \Sigma_\tau$  $f(x)(\sigma,y) \leq g(x)(\sigma,y)$, &\checked &$\times$  &$\times$ & \cite{gp:icalp2014,mp2013:weak-arxiv,brengos2015:jlamp} \\\hline
$\mathcal{F}$ & $\forall x\in X, f(x)\supseteq g(x)$ &\checked &\checked  &\checked &\cite{gahler91,seal09} \\\hline
\end{tabular}
\end{table}
\end{savenotes}

\subsection{Coalgebras and functional simulations}
Assume that a monad $(T,\mu,\eta)$ on $\mathsf{C}$ gives rise to an order-enriched category $\mathcal{K}l(T)$. By $\mathsf{C}_{T,\leq}$ we denote the category whose objects are exactly the objects from $\mathsf{C}_T$ and whose morphisms are oplax homomorphisms. A morphism $f:X\to Y$ in $\mathsf{C}$ is a \emph{oplax homomorphism} between $T$-coalgebras $\alpha:X\to TX$ and $\beta:Y\to TY$ if:

$$
\xymatrix{
X\ar[d]_\alpha \ar[r]^f \ar@{}[dr]|\leq & Y\ar[d]^\beta \\
TX \ar[r]_{Tf} & TY 
}
$$
The inequality in the above diagram can be restated in terms of the composition in $\mathcal{K}l(T)$ by $f^\sharp\circ \alpha \leq \beta\circ f^\sharp$.

 Coalgebras with morphisms satisfying a similar condition  were studied in e.g. \cite{hasuo06} in the context of forward simulations. However, in \emph{loc. cit.} these morphisms are taken from the Kleisli category not the base category. 
 
 These morphisms can be intuitively understood as functional morphisms preserving (and not necessarily reflecting) transitions. In the case of labelled transition systems, the category $\mathsf{Set}_{\mathcal{P}(\Sigma_\tau\times \mathcal{I}d),\leq}$ has all LTS as objects and as morphisms maps between the carriers satisfying the following implication:
$$x\stackrel{a}{\to}_\alpha x' \implies f(x)\stackrel{a}{\to}_\beta f(x') \text{ for any }a\in \Sigma_\tau.$$

Since, as mentioned before, we treat coalgebras with silent moves as endomorphisms in suitable Kleisli categories it is natural to replace $\mathsf{C}_{T,\leq}$ with its endomorphism generalization. Let $J$ be a subcategory of $\mathsf{K}$ with all objects from $\mathsf{K}$. We define $\mathsf{End}_J^\leq (\mathsf{K})$ to be the category of all endomorphisms of $\mathsf{K}$ as objects whose morphisms are given as follows. An arrow $f:X\to Y$ in $J$ is a morphism between $\alpha:X\to X$ and $\beta:Y\to Y$ in $\mathsf{End}_J^\leq (\mathsf{K})$ whenever $f\circ \alpha \leq \beta\circ f$. Whenever $J=\mathsf{K}$ we drop the subscript and write $\mathsf{End}^\leq (\mathsf{K})$ instead of $\mathsf{End}_\mathsf{K}^\leq (\mathsf{K})$. The category 
$\mathsf{End}_J^\leq (\mathsf{K})$ is order enriched with the order on hom-sets directly imposed by the order from $\mathsf{K}$.
\begin{example} If $J=\mathsf{C}$, $\mathsf{K}=\mathcal{K}l(T)$ then $\mathsf{End}_J^\leq (\mathsf{K}) = \mathsf{C}_{T,\leq}$. 
\end{example}

The category $\mathsf{End}_J^\leq (\mathsf{K})$, and therefore also $\mathsf{C}_{T,\leq}$, is a basic ingredient for defining coalgebraic weak bisimilarity. This fact will be justified in the subsection below.

\subsection{Coalgebraic (weak) bisimulation and saturation}\label{subsection:weak_bis_coal} The notions of  strong bisimulation have been well captured coalgebraically \cite{rutten:universal,staton11}. In this paper, we consider Staton's kernel bisimulation \cite{staton11} and instantiate it on single systems only. For a coalgebra $\alpha:X\to TX$  a relation $R\rightrightarrows X$ (i.e. a jointly monic span) in $\mathsf{C}$ is \emph{kernel bisimulation} (or simply \emph{bisimulation})  on $\alpha$ provided that it is a kernel pair of a coalgebraic homomorphism whose domain is $\alpha$. In other words, if there is $\beta:Y\to TY$ and an arrow $f:X\to Y\in \mathsf{C}$ such that $Tf\circ \alpha = \beta \circ f$ for which $R \rightrightarrows X$ is the kernel pair. Since this identity can be restated in terms of the composition in $\mathcal{K}l(T)$ as
$f^\sharp \circ \alpha = \beta \circ f^\sharp$,
we can generalize the definition of bisimulation to the setting of endomorphisms as follows. We say that a relation on $X$ in $J$ is \emph{(strong) bisimulation} on an endomorphism $\alpha:X\to X\in \mathsf{K}$ if it is a kernel pair of an arrow $f:X\to Y \in J$ for which there is $\beta:Y\to Y\in \mathsf{K}$ satisfying $f\circ \alpha = \beta \circ f$. If we take $\mathsf{K}=\mathcal{K}l(T)$ and $J=\mathsf{C}$ then Staton's kernel bisimulation and endomorphism bisimulation coincide.

In \cite{brengos2015:jlamp} we presented a common framework for defining weak bisimulation for coalgebras with internal moves which encompasses several well known instances of this notion for systems among which we find labelled transition systems and fully probabilistic systems. We will now show the basic components of this setting. 

As above, we work in the context of endomorphisms of a category $\mathsf{K}$. However, we additionally assume the following:
\begin{itemize}
\item $\mathsf{K}$ is small,
\item $\mathsf{K}$ is $\omega\mathsf{Cpo}^\vee$-enriched.
\end{itemize}
\begin{remark} \label{remark:smallness_problems} The assumption about smallness is crucial in the construction of a supercategory $\widehat{\mathsf{K}}$ of $\mathsf{K}$ which is $\omega\mathsf{Cpo}^\vee$- enriched and, additionally, left distributive\footnote{See Subsection \ref{subsection:local_embedding} for a detailed description of this construction in the context of $\mathsf{DCpo}^\vee$-enrichment. }. These properties guarantee that, although $\mathsf{K}$ does not always admit saturation, the new category does. However, although, the category $\mathsf{K}=\mathcal{K}l(T)$ is  $\omega\mathsf{Cpo}^\vee$-enriched for all examples of monads considered in this paper, it is never small.
Hence, seemingly this assumption renders the setting useless in our context. As noted in \cite[Rem. 3.2]{brengos2015:jlamp}, there are two potential solutions to the problem. The first solution is to rewrite the whole theory so that not necessarily locally small categories would fit it. Indeed, if the assumption about $\mathsf{K}$ being small is dropped then the hom-objects of $\widehat{\mathsf{K}}$ can form proper partially ordered \emph{classes}. Although the hom-objects of $\widehat{\mathsf{K}}$ would exhibit a $\omega\mathsf{Cpo}^\vee$-like enrichment, formally, this category would not be $\omega\mathsf{Cpo}^\vee$-enriched. The second solution is to take $\mathsf{K}$ to be a suitable full subcategory of $\mathcal{K}l(T)$. For instance, if we focus on a $\mathsf{Set}$-based monad $T$ and $T$-coalgebras whose carrier is of cardinality below $\kappa$ then we can put $\mathsf{K}$ to be the full subcategory of $\mathcal{K}l(T)$ consisting of exactly one set of cardinality $\lambda$ for every $\lambda<\kappa$. For $\kappa=\omega$ this category is dual to the Lawvere theory for $T$ (e.g. \cite{hyland:power:2007}). For the sake of brevity and clarity of the paper we adopt the second solution. Once the conditions to define weak bisimulation are established we will implicitly drop the assumption about smallness of $\mathsf{K}$.
\end{remark}
\begin{definition}\cite{brengos2015:jlamp} We say that a relation $R\rightrightarrows X$ in $J$ is \emph{weak bisimulation} on an endomorphism $\alpha:X\to X$ in $\mathsf{K}$ if it is a kernel pair of a weak behavioural morphism on $\alpha$.
\end{definition}
\noindent In order to complete the above definition we have to present the concept of a weak behavioural morphism. We say that an arrow $f:X\to Y$ in $J$ is \emph{weak behavioural morphism} on $\alpha:X\to X\in \mathsf{K}$ provided that there is an endomorphism $\beta:Y\to Y\in\mathsf{K}$ such that:
\begin{align}
\Theta(\widehat{f}\circ \widehat{\alpha}^\ast) = \Theta ( \widehat{\beta}\circ \widehat{f}). \label{id:WB}
\end{align}
There are several new symbols in the above equation that require an explanation. First of all, $\widehat{(-)}:\mathsf{K}\to \widehat{\mathsf{K}}$ is a locally reflective embedding of $\mathsf{K}$ into a left distributive $\omega\mathsf{Cpo}^\vee$-enriched category $\widehat{\mathsf{K}}$. In \cite{brengos2015:jlamp} we show that such an embedding always exists for $\mathsf{K}$. 
Secondly,  $(-)^\ast:\mathsf{End}^{\leq}(\widehat{\mathsf{K}})\to \mathsf{End}^{\leq\ast}(\widehat{\mathsf{K}})$ arises as the left adjoint in:
\begin{align}
\xymatrix{
\mathsf{End}^{\leq}(\widehat{\mathsf{K}})\ar@/^1pc/[r]^{(-)^\ast} \ar@{}[r]|\perp & \mathsf{End}^{\leq\ast}(\widehat{\mathsf{K}}),
 \ar@/^1pc/[l]^{}
}
\label{adjunction_coalgebraic_saturation} 
\end{align}
where the right adjoint is the inclusion functor from the full subcategory $\mathsf{End}^{\leq\ast}(\widehat{\mathsf{K}})$ of $\mathsf{End}^{\leq}(\widehat{\mathsf{K}})$ whose objects are endomorphisms $\alpha$ additionally satisfying $id\leq \alpha$ and $\alpha \circ \alpha \leq \alpha$.
The adjunction (\ref{adjunction_coalgebraic_saturation}), we refer to as \emph{coalgebraic saturation}, always exists whenever $\widehat{\mathsf{K}}$ is left distributive and $\omega\mathsf{Cpo}^\vee$-enriched \cite[Th. 3.12]{brengos2015:jlamp}. Its left adjoint maps $\widehat{\alpha}:\widehat{X}\to \widehat{X}$ in $\widehat{\mathsf{K}}$ to an endomorphism $\widehat{\alpha}^\ast:\widehat{X}\to \widehat{X}\in \widehat{\mathsf{K}}$ given by $\widehat{\alpha}^\ast = \bigvee_{n\in \mathbb{N}} (id\vee \widehat{\alpha})^n$. The endomorphism $\widehat{\alpha}^\ast$ is the least arrow satisfying $\widehat{\alpha}\leq \widehat{\alpha}^\ast$, $id\leq \widehat{\alpha}^\ast$ and $\widehat{\alpha}^\ast\circ \widehat{\alpha}^\ast\leq \widehat{\alpha}^\ast$. 

The final ingredient in (\ref{id:WB}) that requires an explanation is $\Theta$. Recall that $\widehat{(-)}:\mathsf{K}\to \widehat{\mathsf{K}}$ is locally reflective. This means that for any objects $X,Y\in \mathsf{K}$ the order preserving assignment
$
\widehat{(-)}:\mathsf{K}(X,Y)\to \widehat{\mathsf{K}}(\widehat{X},\widehat{Y})
$
admits a left adjoint. Here, it is denoted by $\Theta_{X,Y}:\widehat{\mathsf{K}}(\widehat{X},\widehat{Y})\to \mathsf{K}(X,Y)$ or simply by $\Theta$ if the subscript objects can be deduced from the context.

Weak behavioural morphisms can be characterized as follows. A given arrow $f:X\to Y$ in $J$ is a weak behavioural morphism on $\alpha:X\to X\in \mathsf{K}$ if and only if there exists $\beta:Y\to Y$ such that:
\begin{align}
\alpha^\ast_f = \beta\circ f, \label{equation:ps}
\end{align}
where $\alpha^\ast_f = \mu x. (f\vee x\circ \alpha)$ is the least fixed point of  $x\mapsto f\vee x\circ \alpha$ \cite{brengos2015:jlamp}.
We will now instantiate this setting on two most prominent examples from \emph{loc. cit.}

\subsubsection{LTS weak bisimulation}\label{subsub:lts_weak_bis}
As mentioned before, the Kleisli category for the LTS monad $\mathcal{P}(\Sigma_\tau\times \mathcal{I}d)$ is $\mathsf{Sup}$-enriched. As a consequence (see \cite[Th. 3.20]{brengos2015:jlamp}):
$$\alpha^\ast_f = f\circ \alpha^\ast,$$ for any labelled transition system  $\alpha:X\to \mathcal{P}(\Sigma_\tau\times X)$ and an arrow $f$ whose domain is $X$ with the composition computed in $\mathcal{K}l(\mathcal{P}(\Sigma_\tau\times \mathcal{I}d))$. Hence, the equation (\ref{equation:ps}) stated in terms of commutativity of a diagram in $\mathsf{Set}$ becomes: 
$$
\xymatrix{
X\ar[d]_{\alpha^\ast} \ar[r]^f & Y \ar[d]^\beta \\
\mathcal{P}(\Sigma_\tau \times X) \ar[r]_{\mathcal{P}(\Sigma_\tau\times f)} & \mathcal{P}(\Sigma_\tau \times Y)
}
$$
for $\alpha:X\to \mathcal{P}(\Sigma_\tau \times X)$, $\beta:Y\to \mathcal{P}(\Sigma_\tau \times Y)$ and $f:X\to Y$, where the transitions of $\alpha^\ast$ are given as follows:
\begin{align*}
& x\stackrel{\tau}{\to}_{\alpha^\ast} x' \iff x (\stackrel{\tau}{\to}_\alpha)^\ast x',\\
& x\stackrel{a}{\to}_{\alpha^\ast} x'\iff x(\stackrel{\tau}{\to}_\alpha)^\ast \circ \stackrel{a}{\to}_\alpha\circ (\stackrel{\tau}{\to}_\alpha)^\ast x'  \text{ for }a\in \Sigma.
\end{align*}
In the above, $S^\ast$ denotes the reflexive and transitive closure of a binary relation $S$.
An equivalence relation $R$ on  $X$ is a weak bisimulation on  $\alpha$ provided that whenever $(x,y)\in R$ we have \cite{brengos2015:lmcs,brengos2015:jlamp}:
\begin{align*}
& x\stackrel{\sigma}{\to}_{\alpha^\ast} x' \implies y\stackrel{\sigma}{\to}_{\alpha^\ast} y' \text{ and }(x',y')\in R,\text{ for any }\sigma\in \Sigma_\tau.
\end{align*}
This coincides with the classical notion of labelled transistion systems weak bisimulation \cite{milner:cc,sangiorgi2011:bis}.
\subsubsection{Fully probabilistic systems weak bisimulation}\label{subsub:fully_prob_weak_bis}
The monad $\mathbb{F}_{[0,\infty]}(\Sigma_\tau\times \mathcal{I}d)$, unlike the LTS monad, does not yield a Kleisli category which is $\mathsf{Sup}$-enriched, or even left distributive. Hence, we cannot simplify (\ref{equation:ps}) as we did for labelled transition systems.  An equivalence relation $R$ on a set $X$ is a weak bisimulation on a system $\alpha:X\to \mathbb{F}_{[0,\infty]}(\Sigma_\tau \times X)$ if and only if the following is satisfied  for any pair $(x,x')\in R$ \cite{brengos2015:jlamp}:
$$
\alpha^\ast_R(x)(\sigma,C)=\alpha^\ast_R(x')(\sigma,C) \text{ for any } \sigma\in \Sigma_\tau \text{ and any abstract class } C \text{ of }R,
$$
where $\alpha^\ast_R:X\to \mathbb{F}_{[0,\infty]} (\Sigma_\tau\times X_{/R})$ is the least solution to:
\begin{align*}
	\alpha^\ast_R(x)(\tau,C) & = 
		f^\sharp(x)(\tau,C)
		\vee
		\sum_{z \in X} \alpha(x)(\tau,z)\cdot \alpha^\ast_R(z)(\tau,C),\\[-3pt]
	\alpha^\ast_R(x)(a,C) & = 
		f^\sharp(x)(a,C)
		\vee
		\sum_{z \in X} \alpha(x)(\tau,z)\cdot \alpha^\ast_R(z)(a,C) + \alpha(x)(a,z)\cdot \alpha^\ast_R(z)(\tau,C).
\end{align*}
In the above, $f:X\to X_{/R}; x\mapsto x_{/R}$. Whenever  $\alpha:X\to \mathbb{F}_{[0,\infty]}(\Sigma_\tau\times X)$ satisfies $\sum_{(\sigma,y)} \alpha(x)(\sigma,y)=~1$ for any $x\in X$ then the above equations reduce to:
\begin{align*}
	\alpha^\ast_R(x)(\tau,C) & = \left \{ \begin{array}{cc} 1 & \text{ if } x\in C,\\
		\sum_{z \in X} \alpha(x)(\tau,z)\cdot \alpha^\ast_R(z)(\tau,C) &\text{ otherwise,} \end{array} \right. 
		\\
	\alpha^\ast_R(x)(a,C) & = 
		\sum_{z \in X} \alpha(x)(\tau,z)\cdot \alpha^\ast_R(z)(a,C) + \alpha(x)(a,z)\cdot \alpha^\ast_R(z)(\tau,C).
\end{align*}
These are exactly the equations considered in \cite{baier97:cav} to define weak bisimulation for fully probabilistic systems. Hence, our coalgebraic notion of weak  bisimulation and Baier and Hermanns' weak bisimulation \cite{baier97:cav} coincide.

\subsection{Relational presheaves}\label{subsection:relational_presheaves}
Lax functors $\mathbb{D}\to \mathsf{Rel} \cong \mathcal{K}l(\mathcal{P})$, known under the name of \emph{relational presheaves}, have been studied in e.g. \cite{niefield2004,sobocinski:jcss}. The motivation for our paper stems from \cite{sobocinski:jcss}, where Soboci\'nski  shows that several examples of systems, among which we find labelled transition systems, tile systems \cite{Gadducci96thetile} and reactive systems \cite{jensen2006} can be modelled as relational presheaves. It is worth noting that the latter two examples are given in terms of relational presheaves whose domain category is not necessarily a one-object category. We refer a curious reader to \cite{sobocinski:jcss}. Here, we only recall the idea proposed by Soboci\'nski to represent labelled transition systems as relational presheaves and encode their saturation in terms of an adjunction.

Any labelled transition system $\alpha: X\to \mathcal{P}(\Sigma_\tau\times X)$ can be viewed as a lax functor $\underline{\alpha}:(\Sigma_\tau)^\ast \to \mathcal{K}l(\mathcal{P})$ given by \cite{sobocinski:jcss}:
\begin{align*}
&\underline{\alpha}(\varepsilon)(x) = \{x\} \text{, }\underline{\alpha}(\sigma)(x) = \{ x' \mid (a,x')\in \alpha(x)\} \text{ for } \sigma \in \Sigma_\tau,\\
&\underline{\alpha}(\sigma_1 \sigma_2\ldots \sigma_n) = \underline{\alpha}(\sigma_1)\circ \underline{\alpha}(\sigma_2)\circ \ldots \circ \underline{\alpha}(\sigma_n) \text{ for }\sigma_i\in \Sigma_\tau.
\end{align*}
Let $[(\Sigma_\tau)^\ast,\mathcal{K}l(\mathcal{P})]$ and $[\Sigma^\ast,\mathcal{K}l(\mathcal{P})]$ denote the categories of relational presheaves on monoid categories $\Sigma_\tau^\ast$ and $\Sigma^\ast$ respectively as objects and oplax transformations as morphisms. In this case, labelled transition systems saturation is \cite{sobocinski:jcss}:
\begin{align}
\xymatrix{
[\Sigma^\ast_\tau,\mathcal{K}l(\mathcal{P})]\ar@/^1pc/[r]^{} \ar@{}[r]|\perp & [\Sigma^\ast,\mathcal{K}l(\mathcal{P})] \ar@/^1pc/[l]^{[p,\mathcal{K}l(\mathcal{P})]}.
}\label{adjunction_saturation} 
\end{align}
The right adjoint is the change-of-base functor:
$$
[p,\mathcal{K}l(\mathcal{P})]:[\Sigma^\ast,\mathcal{K}l(\mathcal{P})]\to [(\Sigma_\tau)^\ast,\mathcal{K}l(\mathcal{P})];\pi \mapsto \pi \circ p,
$$
where $p:(\Sigma_\tau)^\ast \to \Sigma^\ast$ removes all occurences of the letter $\tau$ in words from $(\Sigma_\tau)^\ast$.

 This very nice observation has some limitation. First of all,  although labelled transition systems $X\to \mathcal{P}(M\times X)$ with a monoid structure on labels can be seen as relational presheaves $M\to \mathcal{K}l(\mathcal{P})$ \cite{sobocinski:jcss}, it is not instantly clear how to express their saturation in terms of an adjunction as above. Secondly, although, the choice of $p:(\Sigma_\tau)^\ast\to \Sigma^\ast$ is, to some extent, natural, it does not allow us to see the canonicity in the notion of LTS saturation we see in our approach \cite{brengos2015:lmcs} (i.e. as a reflexive and transitive closure). 
In our opinion, it is only when we encode the label structure inside a monad (cf. Subsection \ref{subsection:coalgebras_with_internal}) and generalize the theory we see the whole picture in which the main role is played by adjunctions similar to (\ref{adjunction_saturation}). Moreover, the new theory becomes consistent with our previous work on weak bisimulation \cite{brengos2015:jlamp,brengos2015:lmcs,brengos2014:cmcs}. 
 

\section{Lax functors}\label{section:lax_functors}
The purpose of this section is to give the definition of a lax functor category and study its properties in the coalgebraic  context generalizing the notion of a relational presheaf. At first, we show that many objects known in mathematics and computer science may be modelled as certain lax functors. Our results generalize both \cite{brengos2015:jlamp} and \cite{sobocinski:jcss} where one can find an extensive list of other examples we do not discuss in this paper.  Secondly, in Subsection \ref{subsection:change_of_base} we study the change-of-base functor between categories of lax functors and its left adjoint. The adjunction forms foundation to the concept of saturation used in the next section to define weak bisimulation for lax functors. 

Throughout this section we assume that:
\begin{itemize}
\item $\mathbb{D}$ is a small category,
\item $\mathsf{K}$ is an order enriched category, 
\item $J$ is a subcategory of $\mathsf{K}$ with all objects from $\mathsf{K}$.
\end{itemize}
Our prototypical example for $J$ and $\mathsf{K}$ are $\mathsf{C}$ and $\mathcal{K}l(T)$ respectively, for a monad $T$ on $\mathsf{C}$. Let $[\mathbb{D},\mathsf{K}]^J$ be the category  whose objects are lax functors from $\mathbb{D}$ to $\mathsf{K}$ and whose morphisms are oplax transformations with components from $J$. Whenever $J=\mathsf{K}$ we will often drop the superscript and write $[\mathbb{D},\mathsf{K}]$ instead of $[\mathbb{D},\mathsf{K}]^\mathsf{K}$. 

The category  $[\mathbb{D},\mathsf{K}]^J$ is order enriched with the order on hom-sets given as follows. For $\pi,\pi'\in [\mathbb{D},\mathsf{K}]^J$  and two oplax transformations $f,f':\pi \to \pi'$ we define:
$$
f\leq f'\iff f_{D} \leq f'_D \text{ in }\mathsf{K} \text{ for any }D\in \mathbb{D}. 
$$

\subsection{Examples and their properties}

We will now describe several examples of the category $[\mathbb{D},\mathsf{K}]^J$ focusing on $\mathbb{D}$ being a monoid category (i.e. a one-object category). In this case, a monoid $M=(M,\cdot,1)$ will be often associated with the one-object category it induces. The only object of the category $M$ will be denoted by $\ast$ and the composition $\circ$ of morphisms $m_1,m_2:\ast\to\ast$ for $m_1,m_2\in M$ given by:
$m_1\circ m_2 = m_1\cdot m_2$.

The first two examples of lax functor categories we consider are categories for which the monoid $M$ in $[M,\mathsf{K}]^J$ is given by:
\begin{itemize}
\item the monoid $\mathbb{N}=(\mathbb{N},+,0)$ of natural numbers with ordinary addition,
\item the one-element monoid $1=(\{0\},+,0)$. 
\end{itemize} 
As will be seen in Section \ref{section:weak_bisimulation_lax_functors}, these two examples play a fundamental role in coalgebraic weak bisimulation.

\subsubsection{The category $[\mathbb{N},\mathsf{K}]^J$}
The purpose of this subsection is to describe the category $[\mathbb{N},\mathsf{K}]^J$ and show the relation between it and the category $\mathsf{End}_J^\leq (\mathsf{K})$. We will show that, intuitively, the objects of the category $[\mathbb{N},\mathsf{K}]^J$ can be understood as (approximations of) iteration of endomorphisms. Indeed, the monoid $\mathbb{N}$ plays the role of a discrete time domain as a lax functor in $[\mathbb{N},\mathsf{K}]^J$ assigns to a given natural number $n$ an approximation of $n$-th power of the given endomorphism.  Although, below we present only one example of this category, it should be noted here that by Corollary \ref{coro:embedding_coalgebras} all coalgebras with internal moves can be seen as lax functors whose domain is $\mathbb{N}$. Apart from labelled transition systems, fully probabilistic systems and filter coalgebras defined in this paper, the reader is referred to e.g. \cite{brengos2015:jlamp} for a long list of other examples of such coalgebras.


For a lax functor  $\pi \in [\mathbb{N},\mathsf{K}]^J$  define $\pi_n = \pi(n):\pi(*)\to \pi(*)$.
Note that any $\pi$ in $[\mathbb{N},\mathsf{K}]^J$ is determined by its sequence $(\pi_n)_{n\in \mathbb{N}}$ and any transformation $f:\pi\to \pi'$ between  $\pi$ and $\pi'$ in $[\mathbb{N},\mathsf{K}]^J$ is determined by the component  $$f_{*}:\pi(*)\to \pi'(*)$$ in $J$. Therefore, for the sake of simplicity of notation, lax functors in $[\mathbb{N},\mathsf{K}]^J$ will be considered as sequences of endomorphisms of $\mathsf{K}$ with a common carrier, and morphisms between lax functors as morphisms between the given carriers in $J$. 

The following three propositions are straightforward to verify and, hence, are left without proofs.
\begin{proposition}\label{proposition_1_T_omega}
A sequence $\pi = (\pi_n)_{n\in \mathbb{N}}$ of endomorphisms in $\mathsf{K}$ with a common carrier is an object of $[\mathbb{N},\mathsf{K}]^J$ if and only if the following conditions are satisfied:
{
\begin{align}
\label{presheaf_property:1} &id\leq \pi_0,\\
\label{presheaf_property:2} &\pi_m\circ \pi_n\leq \pi_{m+n} \text{ for any }m,n\in \mathbb{N}.
\end{align}
}
\end{proposition}

\begin{example}[Labelled transtion systems]\label{example:LTS_presheaves} 
Whenever $J=\mathsf{C}$ and $\mathsf{K}=\mathcal{K}l(T)$ for a monad $T$ on $\mathsf{C}$ then the lax functors of $[\mathbb{N},\mathsf{K}]^J = [\mathbb{N},\mathcal{K}l(T)]^\mathsf{C}$ are sequences of $T$-coalgebras. Let $T= \mathcal{P}^{\Sigma,M}$. A sequence $(\pi_n:X\to \mathcal{P}(\Sigma_\tau \times M \times X))_{n\in \mathbb{N}}$ of labelled transition systems is an object of this category if and only if it satisfies:
\begin{align*}
\infer{x\stackrel{(\tau,1)}{\to}_{\pi_0} x}{}\qquad \infer{x\stackrel{(\sigma,k\cdot l)}{\to}_{\pi_{n+m}} x''}{x\stackrel{(\sigma,k)}{\to}_{\pi_m} x' & x'\stackrel{(\tau,l)}{\to}_{\pi_n} x''} \qquad \infer{x\stackrel{(\sigma,k\cdot l)}{\to}_{\pi_{n+m}} x''}{x\stackrel{(\tau,k)}{\to}_{\pi_m} x' & x'\stackrel{(\sigma,l)}{\to}_{\pi_n} x''} 
\end{align*}

\end{example}

\begin{proposition}\label{proposition_2_T_omega}
Given two lax functors $\pi,\pi'$ in $[\mathbb{N},\mathsf{K}]^J$ an arrow $f:\pi(*)\to \pi'(*)$ in $J$ is a morphism between $\pi$ and $\pi'$  in $[\mathbb{N},\mathsf{K}]^J$ if and only if the following condition holds for all $n\in \mathbb{N}$:
$$
\xymatrix{
\pi(*) \ar[r]^f\ar[d]_{\pi_n}\ar@{}[dr]|\leq & \pi'(*)\ar[d]^{\pi'_n} \\
\pi(*) \ar[r]_{f} & \pi'(*)
}
$$
\end{proposition}

For any endomorphism $\alpha:X\to X\in \mathsf{K}$ define the sequence $\underline{\alpha}=(\underline{\alpha}_n)\in [\mathbb{N},\mathsf{K}]^J$ by  $\underline{\alpha}_n = \alpha^n$. For  $f:X\to Y$ in $J$ which is a morphism  between endomorphisms $\alpha:X\to X$ and $\beta:Y\to Y$ in $\mathsf{End}_{J}^\leq (\mathsf{K})$ put $
\underline{f} = f$.
Clearly, the assignment $\underline{(-)}: \mathsf{End}_{J}^\leq (\mathsf{K})\to [\mathbb{N},\mathsf{K}]^J$ is functorial. Moreover, we have the following.
\begin{proposition}
The functor $\underline{(-)}$ is a full and faithful embedding of the category $ \mathsf{End}_{J}^\leq (\mathsf{K})$ into  $[\mathbb{N},\mathsf{K}]^J$ which preserves the order.
\end{proposition}

Now, consider the functor $(-)_1:[\mathbb{N},\mathsf{K}]^J\to \mathsf{End}_{J}^\leq (\mathsf{K})$ which assigns to any lax functor $\pi=(\pi_n)_{n\in \mathbb{N}}$ the endomorphism $\pi_1$ and  any  morphism $f:\pi\to \pi'$ in $[\mathbb{N},\mathsf{K}]^J$ is assigned to itself. It is clear that this functor preserves the order.
Note that the composition of $\underline{(-)}$ and $(-)_1$ is the identity functor on $\mathsf{End}_{J}^\leq (\mathsf{K})$. 
\begin{proposition}\label{proposition:T_omega_adjunction}
We have the following 2-adjunction:
$$
\xymatrix{
\mathsf{End}_{J}^\leq (\mathsf{K})\ar@/^1pc/[r]^{\underline{(-)}} \ar@{}[r]|\perp & [\mathbb{N},\mathsf{K}]^J \ar@/^1pc/[l]^{{(-)}_1}
}.
$$

\end{proposition}
\begin{proof} The statement follows directly by the fact that for any $\alpha:~X\to X$ in $\mathsf{K}$ and any lax functor  $\pi$ we have
${\mathsf{End}_{J}^\leq (\mathsf{K})}(\alpha,{\pi}_1) = { [\mathbb{N},\mathsf{K}]^J}(\underline{\alpha},\pi)$.  
It is clear that ${\mathsf{End}_{J}^\leq (\mathsf{K})}(\alpha,{\pi}_1) \supseteq { [\mathbb{N},\mathsf{K}]^J}(\underline{\alpha},\pi)$. To see  that the opposite inclusion holds take  $f:X\to Y$ between $\alpha:X\to X$ and $\pi_1:Y\to Y$ in $\mathsf{End}_J ^\leq(\mathsf{K})$. This means that $f\circ \alpha \leq \pi_1 \circ f$. Inductively, we prove $f \circ \alpha^n \leq \pi_1^n \circ f$. Since for a lax functor $\pi\in [\mathbb{N},\mathsf{K}]^J$ we have $\pi_1^n\leq \pi_n$, we directly get that $f \circ \alpha^n \leq \pi_n \circ f$. This proves the assertion.
\end{proof}

As a direct corollary of the above we have:
\begin{corollary}\label{coro:embedding_coalgebras}
For any monad $T$ on $\mathsf{C}$ whose Kleisli category is order enriched we have the following 2-adjunction:
$$
\xymatrix{
\mathsf{C}_{T,\leq} \ar@/^1pc/[r]^{\underline{(-)}} \ar@{}[r]|\perp & [\mathbb{N},\mathcal{K}l(T)]^\mathsf{C} \ar@/^1pc/[l]^{{(-)}_1}
}.
$$
\end{corollary}

\subsubsection{The category $[1,\mathsf{K}]^J$} \label{subsection:monads_lax}

Any lax functor $\pi$ in $[1,\mathsf{K}]^J$ is determined by the underlying endomorphism $\pi(0):\pi(*)\to \pi(*)$ and any oplax transformation is a morphism in $J$ between the carriers of the underlying endomorphisms. Hence, we will identify lax functors and transformations in $[1,\mathsf{K}]^J$ with their endomorphisms and carrier arrows respectively.

\begin{proposition}\label{proposition:t_bullet}
An endomorphism $\pi$ in $\mathsf{K}$ is a lax functor in $[1,\mathsf{K}]^J$ if and only if it satisfies the following:
$$
id\leq  \pi\text{ and }\pi\circ  \pi \leq \pi. 
$$
 An arrow $f$ in $J$ is a transformation between lax functors $\pi$ and $\pi'$ in $[1,\mathsf{K}]^J$ if and only if $f\circ \pi \leq \pi'\circ f$. Hence, $[1,\mathsf{K}]^J$ is isomorphic to $\mathsf{End}^{\leq \ast}_J(\mathsf{K})$.
\end{proposition}
The above result states that objects of $[1,\mathsf{K}]^J$ are exactly reflexive and transitive endomorphisms with morphisms being oplax transformations from $J$. Another way of looking at the category $[1,\mathsf{K}]^J$ is via the following coincidence. There is a one-to-one correspondence between lax functors $1\to \mathsf{K}$ 
and monads \emph{in} $\mathsf{K}$ for an arbitrary 2-category $\mathsf{K}$ \cite{lack:09,lackstreet:2000}. Here, whenever $\mathsf{K}$ is order enriched, the monad unit and the monad multiplication are 2-cells $1\leq \pi$ and
$\pi \circ \pi\leq \pi$ respectively. 

By Proposition \ref{proposition:t_bullet} the category $[1,\mathsf{K}]^J\cong \mathsf{End}_{J}^{\leq,\ast}(\mathsf{K})$ lies at the heart of coalgebraic bisimulation and saturation recalled in Subsection \ref{subsection:weak_bis_coal}. Below we give some examples important from the point of view of coalgebras with internal moves. Note that none of the examples below is connected to fully probabilistic systems. We decide not to include them here, as fully probabilistic system saturation is carried out in a different category (cf. Subsection \ref{subsub:fully_prob_weak_bis}).

\begin{example}[Reflexive and transitive coalgebras] We will now describe objects of $[1,\mathcal{K}l(T)]^\mathsf{C}$ for three examples of $\mathsf{Set}$-based monads $T$: 
\begin{enumerate}
\item $\mathcal{P}$: since $\mathcal{K}l(\mathcal{P})\cong \mathsf{Rel}$, a relation on a given set is a member of the category $[1,\mathcal{K}l(\mathcal{P})]^\mathsf{Set}$ if and only if it is reflexive and transitive, i.e. it is a preorder. 

\item $\mathcal{P}^{\Sigma,M}$: a coalgebra $\alpha:X\to \mathcal{P}(\Sigma_\tau\times M\times X)$ is an object of this category if and only if satisfies
\begin{align*}
\infer{x\stackrel{(\tau,1)}{\to}_{\alpha} x}{}\qquad \infer{x\stackrel{(a,k\cdot l)}{\to}_{\alpha} x''}{x\stackrel{(a,k)}{\to}_{\alpha} x' \text{ \& } x'\stackrel{(\tau,l)}{\to}_{\alpha} x''} \qquad \infer{x\stackrel{(a,k\cdot l)}{\to}_{\alpha} x''}{x\stackrel{(\tau,k)}{\to}_{\alpha} x' \text{ \& } x'\stackrel{(a,l)}{\to}_{\alpha} x''} 
\end{align*}
In the case of $M=1$ these rules are reduced to the rules (\ref{rules_saturated_coalgebra}).

\item $\mathcal{F}$: this category is isomorphic to the category $\mathsf{Top}$ of topological spaces and continuous maps \cite{gahler91,seal_tolen_hoffman14}.
\end{enumerate}
\end{example}

\subsubsection{Coalgebra flows}
The category $[\mathbb{N},\mathcal{K}l(T)]^\mathsf{C}$ can be thought of as a category whose objects represent a single $T$-coalgebra and its (approximations of) finite iteration. In this case, the monoid $\mathbb{N}$ plays the role of a discrete time domain. We can easily replace $\mathbb{N}$ with an arbitrary monoid $M=(M,\cdot,1)$ and generalize Proposition~\ref{proposition_1_T_omega} and \ref{proposition_2_T_omega}. From now on, we associate any lax functor  $\pi \in [M,\mathsf{K}]^J$ with a family $\pi = \{\pi_m:X\to X\}_{m\in M}$ of endomorphisms with a common carrier $X=\pi(\ast)$ which additionally satisfies:
\begin{align}
id_X \leq \pi_1 \text{ and } \pi_{m} \circ \pi_{ n}\leq \pi_{m\cdot n}. \label{prop:weak_flows}
\end{align}
An arrow $f:X\to Y$ in $J$ is a morphism in $[M,\mathsf{K}]^J$ between two lax functors $\pi=\{\pi_m:X\to X\}_{m\in M}$ and $\pi'=\{\pi'_m:Y\to Y\}_{m\in M}$ provided that 
\begin{align*}
f\circ \pi_m \leq \pi_m'\circ  f \text{ for any }m\in M.
\end{align*}
 We call $[M,\mathsf{K}]^J$  the \emph{category of $M$-flows on} $\mathsf{K}$. We will focus on three interesting examples of coalgebra flow categories for $M=[0,\infty)$, $\mathbb{N}\times [0,\infty)$ and $[0,\infty]$.

\begin{example}[Approach spaces]

For certain types of a $\mathsf{Set}$-based monad $T$ and unital quantale $\mathcal{Q}$, the category $[\mathcal{Q},\mathcal{K}l(T)]^\mathsf{Set}$ is of interest from the point of view of generalized topology \cite{seal_tolen_hoffman14,seal09}. The most prominent example is for $\mathcal{Q}$ put to be  the quantale of non-negative real numbers with infinity $[0,\infty]=([0,\infty],+,0,\leq)$ and $T$ taken to be the filter monad $\mathcal{F}$. In this case, the full subcategory of $[[0,\infty],\mathcal{K}l(\mathcal{F})]^\mathsf{Set}$ whose objects additionally satisfy:
$$
\alpha_{\bigvee A} = \bigwedge_{r\in A} \alpha_{a_r} \text{ for any }A\subseteq [0,\infty],
$$
is called \emph{the category of approach spaces} \cite{lowen89,seal09} and is known to naturally extend the categories $\mathsf{Top}$ and $\mathsf{Met}$ (the category of metric spaces and non-expansive maps) \cite{lowen89,seal_tolen_hoffman14}.
\end{example}

\begin{example}[Semantics of timed processes]\label{example:semantics_of_timed}
Timed processes have been defined and studied in \cite{wang90,Larsen199775}. We refer a curious reader to \emph{loc. cit.} for an explicit definition of timed calculus, which we will not recall here, but only focus on its semantics. Let $X$ denote the set of all timed processes. The semantics of timed calculus is given by a labelled transition system $X\to \mathcal{P}([\Sigma_\tau \cup (0,\infty)]\times X)$ whose transitions will be denoted by ${\to}$.
This system can be viewed as a coalgebra $$\alpha:~X\to \mathcal{P}(\Sigma_\tau \times [0,\infty)\times X)$$ given by:
$
{\alpha}(x) = \{ (a,0,x') \mid x\stackrel{a}{\to} x',a\in \Sigma_\tau\} \cup \{(\tau,r,x') \mid x\stackrel{r}{\to} x', r\in (0,\infty)\}.
$
The main purpose for the change of the type functor is to put the original system into the setting of coalgebras whose type is a monad. Indeed, since $[0,\infty)$ with ordinary addition is a monoid, the functor $\mathcal{P}(\Sigma_\tau\times [0,\infty)\times \mathcal{I}d)=\mathcal{P}^{\Sigma,[0,\infty)}$ is a monad as in Subsection~\ref{subsection:lts_monads}. Hence,
by Corollary \ref{coro:embedding_coalgebras} we can consider $\underline{{\alpha}}=(\alpha^n)_{n\in \mathbb{N}}$, a member of the category $[\mathbb{N},\mathcal{K}l(\mathcal{P}^{\Sigma,[0,\infty)})]^\mathsf{Set}$. The following observation will allow us to slightly change the perspective on these systems.
\begin{theorem}\label{theorem:change_of_monoids_lts_monad}
For any monoids $M$ and $M'$ we have:
$$
[M',\mathcal{K}l(\mathcal{P}^{\Sigma,M})]^\mathsf{Set} \cong [M' \times M, \mathcal{K}l(\mathcal{P}^{\Sigma,1})]^\mathsf{Set}.
$$
\end{theorem}
\begin{proof} Follows easily by the sequence of bijective correspondences below:
$$
\infer{\infer{\infer{\infer{M'\times M \to \mathcal{P}(\Sigma_\tau \times X)^X}{X\times M'\times M \to \mathcal{P}(\Sigma_\tau \times X)}}{X\times M' \to \mathcal{P}(\Sigma_\tau \times X)^M}}{X\times M' \to \mathcal{P}(\Sigma_\tau \times M\times X)}}{{M' \to \mathcal{P}(\Sigma_\tau \times M\times X)^{X}}}
$$
\end{proof}
\noindent In particular, this means that $[\mathbb{N},\mathcal{K}l(\mathcal{P}^{\Sigma,[0,\infty)})]^\mathsf{Set} \cong [\mathbb{N}\times [0,\infty),\mathcal{K}l(\mathcal{P}^{\Sigma,1})]^\mathsf{Set}$. Thus, $\underline{\alpha}$ can  be viewed as a lax functor $\mathbb{N}\times [0,\infty)\to \mathcal{K}l(\mathcal{P}(\Sigma_\tau\times \mathcal{I}d))$ which maps any pair $(n,t)\in \mathbb{N}\times [0,\infty)$ to a coalgebra $X\to \mathcal{P}(\Sigma_\tau\times X)$ given by:
\begin{align*}
&x\mapsto 
\{(\sigma,x') \mid x\stackrel{(\tau,t_1)}{\to}_\alpha\circ \ldots \circ \stackrel{(\tau,t_{k-1})}{\to}_\alpha\circ \stackrel{(\sigma,t_k)}{\to}_\alpha\circ \stackrel{(\tau,t_{k+1})}{\to}_\alpha\circ \ldots \circ \stackrel{(\tau,t_{n})}{\to}_\alpha x'\},
\end{align*}
where $t_1+\ldots +t_n =t$. We will now derive two new coalgebras from $\alpha$, namely $\alpha^\ast$ and $\alpha^T$,  and discuss their properties. These two coalgebras will play a crucial role in modelling different types of behavioural equivalences on $\alpha$ (see Example \ref{example:timed_systems_weak} for details).

 Define $\alpha^\ast:X\to \mathcal{P}(\Sigma_\tau\times [0,\infty)\times X)$ whose transitions are:
\begin{align*}
& x\stackrel{(\tau,t)}{\to}_{\alpha^\ast} x' \text{ iff } x \stackrel{(\tau,t_1)}{\to}_{\alpha} \circ \ldots \circ \stackrel{(\tau,t_n)}{\to}_{\alpha}   x' \text{ for }t=\sum_{i=1}^{n} t_i,\\
& x\stackrel{(a,t)}{\to}_{\alpha^\ast} x' \text{ iff } x\stackrel{(\tau,t_1)}{\to}_{\alpha^\ast} x''\stackrel{(a,t_2)}{\to}_{\alpha}x''' \stackrel{(\tau,t_3)}{\to}_{\alpha^\ast} x' \text{ for }a\in \Sigma \text{ and } t=t_1+t_2+t_3.
\end{align*}
The coalgebra $\alpha^\ast$,  viewed as an endomorphism in $\mathcal{K}l(\mathcal{P}^{\Sigma,[0,\infty)})$, is an object of $[1,\mathcal{K}l(\mathcal{P}^{\Sigma,[0,\infty)})]^\mathsf{Set}$. However, since by Theorem \ref{theorem:change_of_monoids_lts_monad} we have $$[1,\mathcal{K}l(\mathcal{P}^{\Sigma,[0,\infty)})]^\mathsf{Set} \cong [[0,\infty), \mathcal{K}l(\mathcal{P}^{\Sigma,1})]^\mathsf{Set},$$ we can also view $\alpha^\ast$ as a lax functor $[0,\infty)\to \mathcal{K}l(\mathcal{P}(\Sigma_\tau\times \mathcal{I}d))$. It maps any $t\in [0,\infty)$ to the coalgebra
$$X\to \mathcal{P}(\Sigma_\tau\times X); x\mapsto \{ (\sigma,x')\mid (\sigma,t,x') \in \alpha^\ast(x)\}.$$
Finally, let $\alpha^T:X\to \mathcal{P}(\Sigma_\tau\times X)$ be the labelled transition system whose transitions are defined by:
$$
x\stackrel{\sigma}{\to}_{\alpha^T} x'\iff x\stackrel{(\sigma,t)}{\to}_{\alpha^\ast} x' \text{ for some }t\in [0,\infty).
$$
A straight forward verification proves that $\alpha^T$ is an object in $[1,\mathcal{K}l(\mathcal{P}(\Sigma_\tau\times \mathcal{I}d))]^\mathsf{Set}$. 

As mentioned before, Example \ref{example:timed_systems_weak} will complete the whole picture on timed processes semantics. We will show that whenever we consider $\underline{\alpha}$ as a lax functor $\mathbb{N}\to\mathcal{K}l(\mathcal{P}^{\Sigma,[0,\infty)})$, weak bisimulation on $\underline{\alpha}$ is the so-called weak timed bisimulation. However, if $\underline{\alpha}$ is viewed as a lax functor $\mathbb{N}\times [0,\infty)\to \mathcal{K}l(\mathcal{P}(\Sigma_\tau\times \mathcal{I}d))$ then its weak bisimulation becomes weak time-abstract bisimulation \cite{Larsen199775}.
\end{example}

\begin{example}[Transition of a continuous time Markov chain] \label{example:markov_chain_transition} The purpose of this example is to show that the transition matrix of a continuous time Markov chain, or CTMC in short, may be viewed as a lax functor $[0,\infty)\to\mathcal{K}l(\mathbb{F}_{[0,\infty]})$. 
%
Here, we only recall some notions from Markov chain theory. The reader is referred to e.g. \cite{books/daglib/0095301} for basic definitions and properties.

Let $(X_t)_{t\geq 0}$ be a CTMC. We call the chain $(X_t)_{t\geq 0}$ \emph{homogeneous} whenever $\mathbb{P}(X_t = j \mid X_s=i) = \mathbb{P}(X_{t-s}=j \mid X_0=i)$.
Any homogeneous CTMC $(X_t)_{t\geq 0}$ on an at most countable state space $S$ gives rise to its \emph{transition matrix}, i.e. a family  $\{P(t):S^2\to [0,1]\}_{t\geq 0}$ whose $ij$-th entry $p_{ij}(t)=P(t)(i,j)$ describes the conditional \emph{transition probabilities}:
$$
p_{ij}(t) = \mathbb{P}(X_t = j \mid X_0=i).
$$ 
The transition matrix satisfies $P(0)=I$ and $P(t+s)=P(t)\cdot P(s)$, where $I$ is the identity matrix and $\cdot$ is the matrix multiplication. The transition matrix $\{P(t)\}_{t\geq 0}$ yields an assignment $\pi:[0,\infty)\to \mathcal{K}l(\mathbb{F}_{[0,\infty]})$ given for any $t\in [0,\infty)$ by:
$$
\pi(\ast) = S, \quad \pi_t:S\to \mathbb{F}_{[0,\infty]} S; \pi_t(i)(j)= p_{ij}(t).
$$
The assignment $\pi = (\pi_t)_{t\in [0,\infty)}$ is a strict functor $[0,\infty)\to \mathcal{K}l(\mathbb{F}_{[0,\infty]})$ and, hence, is a member of $[[0,\infty),\mathcal{K}l(\mathbb{F}_{[0,\infty]})]^\mathsf{Set}$ and will be referred to as a \emph{transition functor} of the homogeneous chain $(X_t)_{t\geq 0}$. We elaborate more on transition functors and their weak bisimulation in Example \ref{example:markov_weak_bisim}.
\end{example}

\subsection{Change-of-base functor and its left adjoint}\label{subsection:change_of_base}
 Before we state the definition of weak bisimulation on lax functors we need one technical result regarding the change-of-base functor and existence of its left adjoint.
 
Any functor $p:\mathbb{D}\to \mathbb{E}$ between small categories $\mathbb{D}$ and $\mathbb{E}$ yields a functor $$[p,\mathsf{K}]^J:[\mathbb{E},\mathsf{K}]^J\to [\mathbb{D},\mathsf{K}]^J$$ defined as follows. For any $\pi\in [\mathbb{E},\mathsf{K}]^J$ put $ [p,\mathsf{K}]^J(\pi) = \pi\circ p$ and for any oplax transformation $f=\{f_E\}_{E\in \mathbb{E}}$ between $\pi$ and $\pi'$ in $[\mathbb{E},\mathsf{K}]^J$ 
the $D$-component of  $[p,\mathsf{K}]^J (f):\pi\circ p \to \pi'\circ p$ is given by  $f_{p D} = \pi(p D) \to \pi'(p D)$. In other words:
$$
 [p,\mathsf{K}]^J(f)_D = f_{p D} .
$$ 
It is easy to check that $[p,\mathsf{K}]^J$ is locally monotonic. 

\subsubsection{The general case}\label{subsection:the_general_case}
In this paragraph we assume the following:
\begin{itemize}
\item $p:\mathbb{D}\to \mathbb{E}$ is a functor between small categories,
\item $J$ and $\mathsf{K}$ have all small coproducts and the inclusion functor $J \hookrightarrow \mathsf{K}$ preserves them,
\item  $\mathsf{K}$ is $\mathsf{Sup}$-enriched (i.e. $\mathsf{K}$ is a \emph{quantaloid} \cite{rosen:quanta96}),
\item all suprema in hom-sets of $\mathsf{K}$ are preserved by arbitrary cotupling, i.e. 
$$
[\{\bigvee_{i_j} f_{i_j}\}_j ] = \bigvee_{j,i_j} [\{ f_{i_j}\}_j ].
$$
\end{itemize}

\begin{example}\label{example:monads-general-assump}
The above assumptions are true for $\mathsf{K}=\mathcal{K}l(\mathcal{Q}^{(-)})$ for an arbitrary quantale $\mathcal{Q}$ and $J=\mathsf{Set}$ or $J=\mathsf{K}$. Hence, in particular, for $\mathsf{K}= \mathcal{K}l(\mathcal{P}^{\Sigma,M})$.
\end{example}

\begin{theorem}\label{theorem:left_adjoint}
The functor $[p,\mathsf{K}]^J:[\mathbb{E},\mathsf{K}]^J\to [\mathbb{D},\mathsf{K}]^J$ admits a left $2$-adjoint $\Sigma_p$.
\end{theorem}
\begin{proof} The proof of this theorem is divided into two parts. In the first part we present an assignment $\Sigma_p$ and show it is a well defined functor between suitable categories. In the last part we show that $\Sigma_p$ is a left $2$-adjoint to the functor $[p,\mathsf{K}]^J:[\mathbb{E},\mathsf{K}]^J\to [\mathbb{D},\mathsf{K}]^J$.

\noindent \textbf{Part 1.} For any two objects $X,Y$ in $\mathsf{K}$ let $\perp_{X,Y}$ denote the least element in the poset $\mathsf{K}(X,Y)$. We will often drop the subscript and write $\perp$ instead.
For any $\pi\in [\mathbb{D},\mathsf{K}]^J$ define an assignment $\Sigma_p(\pi)$ from the category $\mathbb{E}$ to $\mathsf{K}$ on an object $E\in \mathbb{E}$ and a morphism $e:E_1\to E_2\in \mathbb{E}$ by:
\begin{align*}
\Sigma_p(\pi)(E) = \sum_{{D: pD = E}}\pi D \text{ and }\Sigma_p(\pi)(e) = \bigvee_{d: pd = e} \overline{\pi d}, 
\end{align*}
where $\overline{\pi d}:\Sigma_p(\pi)(E_1)\to \Sigma_p(\pi)(E_2)$ is given as follows. Let $d:D_1\to D_2$ and $pD_1 = E_1$, $pD_2=E_2$. We define the morphism  $\overline{\pi d}$ via cotupling in $\mathsf{K}$ by:
\begin{align*}
&\overline{\pi d} = [\{ \delta_{D}\}_{D: pD = E_1} ], \text{ where } \\
&\delta_D = \left \{ \begin{array}{ccc}
\mathsf{in}_{\pi D_2} \circ \pi d &:\pi D_1\to \sum_{D':pD'=E_2}\pi D' & \text{ if $D=D_1$},\\
\perp & :\pi D\to \sum_{D':pD'=E_2}\pi D' & \text{ otherwise.}
\end{array} \right.
\end{align*} 
We will now show that for any $\pi \in[\mathbb{D},\mathsf{K}]^J$ the assignment $\Sigma_{p}(\pi):\mathbb{E}\to \mathsf{K}$ is a lax functor. Indeed, take $id_E:E\to E$ in $\mathbb{E}$. We have:
$
\Sigma_p(\pi)(id_E) = \bigvee_{pd=id_E} \overline{\pi d}$. There can be two cases. If there is no $D$ mapped onto $E$ by the functor $p$ then $\Sigma_p(\pi)(E)$ is the initial object in $\mathsf{K}$. In this case the identity morphism on $\Sigma_p(\pi)(E)$ and the morphism $\Sigma_p(\pi)(id_E)$ are both equal to the least morphism $\perp$. Now, for any object $D$ such that $pD = E$ we have 
\begin{itemize}
\item $p(id_D) = id_E$, 
\item $\pi(id_D) \geq id_{\pi D}$.
\end{itemize} 
By the fact that cotupling preserves suprema we get 
{
$$\Sigma(\pi)(id_E)=\bigvee_{d:pd=id_E} \hspace{-0.3cm} \overline{\pi d} \geq \hspace{-0.3cm}\bigvee_{D: pD=E} \hspace{-0.3cm} \overline{\pi(id_D)}\geq \hspace{-0.3cm} \bigvee_{D: pD=E} \hspace{-0.3cm} \overline{id_{\pi D}} = id_{\Sigma_p(\pi)(E)}.$$} 

Now take $E_1\stackrel{e}{\to} E_2\stackrel{e'}{\to} E_3$ in $\mathbb{E}$. We have:
{
\begin{align*}
& \Sigma_p(\pi)(e'\circ e) = \bigvee_{d: p d = e'\circ e}\hspace{-0.3cm} \overline{\pi d}\geq  \bigvee_{pd_1=e',pd_2=e \text{ and }d_1,d_2 \text{ are composable}}\hspace{-2cm}\overline{\pi(d_1\circ d_2)}\geq \\
&\bigvee_{pd_1=e',pd_2=e \text{ and }d_1,d_2 \text{ are composable}}\hspace{-2cm}\overline{\pi(d_1)\circ  \pi(d_2)} \stackrel{\diamond}{=} \bigvee_{pd_1=e',pd_2=e} \hspace{-0.5cm}\overline{\pi(d_1)} \circ \overline{\pi (d_2)} =\\
&\bigvee_{pd_1=e'}\overline{\pi (d_1)} \circ \bigvee_{pd_2=e}\overline{\pi(d_2)}   = \Sigma_p(\pi)(e')\circ \Sigma_p(\pi)(e).
\end{align*}
}
\noindent The equation marked with $(\diamond)$ requires some explanation. If $d_1$ and $d_2$ are composable then $\overline{\pi(d_2)\circ  \pi(d_1)} =\overline{\pi(d_2)}\circ  \overline{\pi(d_1)}.$  If they are not composable in $\mathbb{D}$ then $\overline{\pi(d_2)}\circ \overline{\pi(d_1)} = \perp$, so clearly this equation holds.

For any oplax transformation $f:~\pi \to \pi'$ between $\pi,\pi'$ in $[\mathbb{D},\mathsf{K}]^J$ we put $\Sigma_p(f):\Sigma_p(\pi)\to \Sigma_p(\pi')$ whose $E$-component is given by:
$$
\Sigma_p(f)_E = \sum_{_{pD = E}}\pi D \stackrel{\sum_D f _D }{\to} \sum_{_{pD = E}}\pi' D. 
$$
Note that the $E$-component of $\Sigma_p(f)$ comes from the base category $J$. This follows by our assumptions about the inclusion functor $J \hookrightarrow \mathsf{K}$ preserving all small coproducts and the fact that all components of $f$ are arrows in $J$. It is clear that $\Sigma_p$ is functorial. 
This part of the proof is now completed.

\noindent \textbf{Part 2.} We will now prove that $\Sigma_p$ is a left 2-adjoint to $[p,\mathsf{K}]^J$. Here we should note that the remaining part of the proof is almost the same as the proof of a similar statement concerning relational persheaves \cite{sobocinski:jcss}. For any  $\pi\in [\mathbb{D},\mathsf{K}]^J$ define a transformation $\eta_\pi:\pi \to [p,\mathsf{K}]^J(\Sigma_p(\pi)) = \Sigma_p(\pi)\circ p$  whose $D$-component is given by the coprojection into the component of the coproduct indexed with $D$:
$$
(\eta_\pi)_D :\pi D \to \Sigma_p(\pi)(pD)=\sum_{D': pD'=pD}\pi D'; \quad  (\eta_\pi)_D = \mathsf{in}_{\pi D}.
$$
We have the following:
\begin{enumerate}[(a)]
\item since $J$ and $\mathsf{K}$ have all small coproducts and since $J\hookrightarrow \mathsf{K}$ preserves these coproducts the $D$-component of $\eta_\pi$ comes from the underlying category $J$;\label{property:eta_a} 
\item $\eta$ is an oplax transformation between lax functors $\pi$ and $\Sigma_p(\pi)\circ p$ in $[\mathbb{D},\mathsf{K}]^J$. \label{property:eta_b} To see this consider any $d:D_1\to D_2$ and note that:
{
$$\xymatrix{
\pi D_1 \ar@{}[drr]|\leq \ar[d]_{\pi d} \ar[rr]^{\mathsf{in}_{\pi D_1}} &  & \sum_{D':pD'=pD_1} \pi D' \ar[d]^{\Sigma_p(\pi)(pd)=\bigvee_{d':pd' = pd} \overline{\pi d'}}\\
\pi D_2 \ar[rr]_{\mathsf{in}_{\pi D_2}} & & \sum_{D':pD'=pD_2} \pi D'
}
$$
}
\item $\eta$ is a natural transformation from the identity functor $[\mathbb{D},\mathsf{K}]^J\to [\mathbb{D},\mathsf{K}]^J$ to the functor $[p,\mathsf{K}]^J\circ \Sigma_p$.
\end{enumerate}
We will check that $\eta$ satisfies the universal property of units. Consider any transformation $f:\pi \to T^p(\pi')=\pi'\circ p$ in $[\mathbb{D},\mathsf{K}]^J$ for $\pi'\in [\mathbb{E},\mathsf{K}]^J$. By the universal properties of the coproduct for any $E\in \mathbb{E}$ there is a unique morphism 
$g_E:\sum_{D: pD=E} \pi D \to ~\pi'(E)$ in $\mathsf{K}$  for which the following diagram commutes:
{$$
\xymatrix{
\pi D \ar[rrr]^{(\eta_\pi)_D=\mathsf{in}_{\pi D}} \ar[drrr]_{f_D} & & & \sum_{ pD'=pD}\pi D' \ar@{-->}[d]^{g_{pD}}\\
&&&  \pi'(pD)
}
$$
}
\noindent By (\ref{property:eta_a}) we directly see that $g_E$ is a morphism in $J$. In order to complete the proof we need to show that the family $g = (g_E)_{E\in \mathsf{E}}$ is a transformation from $\Sigma_p(\pi)$ and $\pi'$ in $[\mathbb{E},\mathsf{K}]^J$. We need to show that for any $e:E\to E'$ in $\mathbb{E}$ we have:
{ 
$$
\xymatrix{
\Sigma_p(\pi)(E) \ar@{}[dr]|\leq \ar[d]_{\Sigma_p(\pi)(e) = \bigvee_{d:pd = e} \overline{\pi d} }\ar[r]^{g_E} & \pi'(E)\ar[d]^{\pi'(e)} \\
\Sigma_p(\pi)(E') \ar[r]_{g_{E'}} & \pi'(E')
}
$$
}
Clearly, it is enough if we focus on morphisms from $\mathbb{E}$ which are images of morphisms from $\mathbb{D}$ under $p$. Indeed, if $e$ is not of this form then the diagram above lax commutes as $\Sigma_p(\pi)(e) = \perp$. By our assumptions about $f$ and by (\ref{property:eta_b})  the front square and the parallelogram on the back in the diagram below lax commute for arbitrary $d:D\to D'$ in $\mathbb{D}$. By the fact that cotupling preserves all suprema the parallelogram on the right also lax commutes. 
{
$$
\xymatrix{
& & \sum_{D'':pD'' = pD} \pi D'' \ar[dl]^{g_{p D}} \ar[d]^{\bigvee_{d:pd = e} \overline{\pi d}} \\
\pi D \ar[urr]^{(\eta_\pi)_D=\mathsf{in}_{\pi D}} \ar[d]_{\pi d}\ar@{}[dr]|\leq  \ar[r]_{f_D} & \pi'(p D) \ar[d]|{\pi'(e)} & \sum_{D'':pD'' = pD'} \pi D'' \ar@/^1pc/[dl]^{g_{p D'}} \\
\pi D'\ar@{-->}[urr]_{\mathsf{in}_{\pi D'} } \ar[r]_{f_D} & \pi'(p D')
}
$$
}
This completes the proof.
\end{proof}

The theorem above encompasses results presented in  \cite{niefield2004} (for $J=\mathsf{Set}$ and  $\mathsf{K}=\mathcal{K}l(\mathcal{P})$) and \cite{sobocinski:jcss} (for $J=\mathsf{K}=\mathcal{K}l(\mathcal{P})$) and, in the light of Example \ref{example:monads-general-assump}, it may be viewed as a generalization of these statements to $\mathcal{K}l(\mathcal{Q}^{(-)})$-valued lax functors.

\subsubsection{The adjunction $[\mathbb{D},\mathsf{K}]\rightleftarrows [1,\mathsf{K}]$}
The most important adjunction between lax functor categories from the point of view of weak bisimulation is the adjunction $[\mathbb{D},\mathsf{K}]\rightleftarrows [1,\mathsf{K}]$. In this case, the restrictive assumptions from the previous paragraph can be relaxed. Here, we assume that:
\begin{itemize}
\item $J$ and $\mathsf{K}$ have arbitrary coproducts of families indexed by objects from $\mathbb{D}$ and the inclusion functor $J \hookrightarrow \mathsf{K}$ preserves them,
\item  $\mathsf{K}$ is left distributive $\mathsf{DCpo}^\vee$-enriched,
\item  cotupling preserves the order, i.e. if $f_i\leq g_i :X_i\to Y$ for any $i\in I$ with $|I|\leq |\mathbb{D}|$ then:
$$
[f_i] \leq [g_i].
$$
\end{itemize}

 In this paragraph we will prove that given the above conditions the change-of-base functor $[!,\mathsf{K}]^J:[1,\mathsf{K}]^J\to [\mathbb{D},\mathsf{K}]^J$ admits a left $2$-adjoint $\Sigma_!$. However, before we do this, we need to define some ingredients necessary to derive its formula.

Assume $\pi\in [\mathbb{D},\mathsf{K}]^J$ is a lax functor and let $d:D_1\to D_2$ be a morphism in $\mathbb{D}$. We define an endomorphism $\overline{\pi(d)}:\sum_{D\in \mathbb{D}} \pi(D)\to \sum_{D\in \mathbb{D}} \pi(D)$ in $\mathsf{K}$ given by $\overline{\pi(d)} = [\delta_{D}]_{D\in \mathbb{D}}$, where 
\begin{align*}
&\delta_D = \left \{ \begin{array}{ccc}
\mathsf{in}_{\pi D_2} \circ \pi (d) &:\pi D_1\to \sum_{D'\in \mathbb{D}}\pi D' & \text{ if $D=D_1$},\\
\mathsf{in}_{\pi D} & :\pi D\to \sum_{D'\in\mathbb{D}}\pi D' & \text{ otherwise.}
\end{array} \right.
\end{align*} 
Since cotupling preserves the order, we have $id\leq \overline{\pi(id_D)}$ for any $D\in \mathbb{D}$. Note that it is not necessary to assume $\mathsf{K}$ is left distributive $\mathsf{DCpo}^\vee$-enriched in order to construct $\overline{\pi(d)}$. This observation will be used in the next section, where we work with $\overline{\pi(d)}$ even though $\mathsf{K}$ does not satisfy this property.

Finally, consider:
$$
\Pi = \{ \overline{\pi(d_1)} \vee \ldots \vee \overline{\pi(d_k)} \mid k\in \{1,2,\ldots \} \text{ and } d_i\text{ is a morphism in } \mathbb{D}\}. 
$$
We are now ready to define $\Sigma_!(\pi):\sum_{D\in \mathbb{D}} \pi(D)\to \sum_{D\in \mathbb{D}} \pi(D)$. We put:
$$
\Sigma_!(\pi) = \bigvee_{l\in \mathbb{N}} (\bigvee \Pi)^l,
$$ 
Note that:
\begin{align*}
& id \leq \overline{\pi(id_D)}\leq \bigvee \Pi \leq \Sigma_!(\pi) \text{ for any }D\in \mathbb{D}, \\
& \Sigma_!(\pi)\circ \Sigma_!(\pi) = \bigvee_{l\in \mathbb{N}} (\bigvee \Pi)^l \circ \bigvee_{l\in \mathbb{N}} (\bigvee\Pi)^l = \bigvee_{l_1,l_2\in \mathbb{N}} (\bigvee\Pi)^{l_1+l_2} =\Sigma_!(\pi).
\end{align*}
Hence, $\Sigma_!(\pi)$ is an object of $[1,\mathsf{K}]^J$. Now, for any oplax transformation $f:\pi\to \pi'$ in $[\mathbb{D},\mathsf{K}]^J$ put $\Sigma_!(f) = \sum_{D\in \mathbb{D}}f_D$.
\begin{lemma}
 $\Sigma_!:[\mathbb{D},\mathsf{K}]^J\to [1,\mathsf{K}]^J$ is a locally monotonic functor. 
\end{lemma}
\begin{proof}
In order to prove the statement it is enough to show that $\Sigma_!(f)$ is an oplax transformation between $\Sigma_!(\pi)$ and $\Sigma_!(\pi')$ whenever $f$ is an oplax transformation between $\pi$ and $\pi'$ in $[\mathbb{D},\mathsf{K}]^J$. Thus we have:
\begin{align*}
&\Sigma_!(f)\circ \Sigma_!(\pi) = \Sigma_!(f) \circ \bigvee_{l\in \mathbb{N}}  (\bigvee \Pi)^l =  \bigvee_{l\in \mathbb{N}} \Sigma_!(f) \circ (\bigvee \Pi)^l \stackrel{\dagger}{\leq} \\
&\bigvee_{l\in \mathbb{N}} (\bigvee \Pi')^l \circ \Sigma_!(f) = \Sigma_!(\pi')\circ \Sigma_!(f).
\end{align*}
The inequality marked with $(\dagger)$ follows by $\Sigma_!(f) \circ (\bigvee \Pi)^l\leq (\bigvee \Pi')^l \circ \Sigma_!(f)$ which is proved inductively. For $l=1$ we have:
\begin{align*}
&\Sigma_!(f) \circ (\bigvee \Pi) = \bigvee \Sigma_!(f)\circ \Pi = \\
&\bigvee  \{ \Sigma_!(f)\circ (\overline{\pi(d_1)} \vee \ldots \vee \overline{\pi(d_k)}) \mid k\in \{1,2,\ldots \} \text{ and } d_i\text{ is a morphism in } \mathbb{D}\} = \\
&\bigvee  \{ \Sigma_!(f)\circ \overline{\pi(d_1)} \vee \ldots \vee\Sigma_!(f)\circ \overline{\pi(d_k)} \mid k\in \{1,2,\ldots \} \text{ and } d_i\in \mathbb{D}\} \stackrel{\dagger\dagger}{\leq}\\
&\bigvee  \{ \overline{\pi'(d_1)}\circ \Sigma_!(f) \vee \ldots \vee \overline{\pi'(d_k)}\circ \Sigma_!(f) \mid k\in \{1,2,\ldots \} \text{ and } d_i\in \mathbb{D}\}\leq\\
&\bigvee  \{ (\overline{\pi'(d_1)} \vee \ldots \vee \overline{\pi'(d_k)})\circ \Sigma_!(f) \mid k\in \{1,2,\ldots \} \text{ and } d_i\in \mathbb{D}\} = (\bigvee \Pi')\circ \Sigma_!(f).
\end{align*} 
The inequality marked with $(\dagger\dagger)$ is true since for any morphism $d:D_1\to D_2$ in $\mathbb{D}$ we have $\Sigma_!(f)\circ \overline{\pi(d)}\leq \overline{\pi'(d)}\circ \Sigma_!(f)$. This is a consequence of the fact that $f$ is an oplax natural transformation between $\pi$ and $\pi'$ and that the order is preserved by cotupling.
\end{proof}

\begin{theorem}\label{theorem:dcpo_left_adjoint}
The functor $\Sigma_!:[\mathbb{D},\mathsf{K}]^J\to [1,\mathsf{K}]^J$ is a left $2$-adjoint to the change-of-base functor $[!,\mathsf{K}]:[1,\mathsf{K}]^J \to [\mathbb{D},\mathsf{K}]^J$.
\end{theorem}
\begin{proof}
We will prove that for any lax functor $\pi:\mathbb{D}\to \mathsf{K}$ and $\pi':1\to \mathsf{K}$ the partially ordered hom-sets  $[\mathbb{D},\mathsf{K}]^J(\pi,[!,\mathsf{K}](\pi'))$ and $[1,\mathsf{K}]^J(\Sigma_!(\pi),\pi')$ are isomorphic. Take an oplax transformation $f:\pi\to[!,\mathsf{K}](\pi')=\pi'\circ !$. We have the following sequence of equivalent statements:
\begin{align}
&f_{D_2}\circ \pi(d) \leq (\pi'\circ !)(d)\circ f_{D_1} = \pi'\circ f_{D_1} \text{ for any }d:D_1\to D_2\in \mathbb{D}\label{ineq:1},\\
&[f_D]\circ \overline{\pi(d)} \leq  \pi'\circ [f_{D}] \text{ for any }d:D_1\to D_2\in \mathbb{D} \label{ineq:2},\\
&\bigvee [f_D]\circ \Pi \leq  \pi'\circ [f_{D}],\label{ineq:3.5}\\
&[f_D]\circ \bigvee \Pi \leq  \pi'\circ [f_{D}]\label{ineq:3},\\
&[f_D]\circ (\bigvee \Pi)^l \leq  \pi'\circ [f_{D}] \text{ for any }l\in \mathbb{N}\label{ineq:4},\\
&\bigvee_{l\in \mathbb{N}}[f_D]\circ (\bigvee \Pi)^l \leq  \pi'\circ [f_{D}],\\
&[f_D]\circ \Sigma_!(\pi) \leq  \pi'\circ [f_{D}].
\end{align}
The implication (\ref{ineq:1})$\implies$ (\ref{ineq:2}) is a consequence of the fact that cotupling preserves the order and $id\leq \pi'$ (hence, $f_D\leq \pi'\circ f_D$ for any $D\in \mathbb{D}$).  (\ref{ineq:2}$\iff$\ref{ineq:3.5}) follows by left distributivity of $\mathsf{K}$. The implication  (\ref{ineq:3})$\implies$ (\ref{ineq:4}) follows by induction and $\pi'\circ \pi'\leq \pi'$. Therefore, the isomorphism  between $[\mathbb{D},\mathsf{K}]^J(\pi,[!,\mathsf{K}](\pi'))$ and $[1,\mathsf{K}]^J(\Sigma_!(\pi),\pi')$ is given by $f\mapsto [f_D]$. It is natural in $\pi$ and $\pi'$ which ends the proof.
\end{proof}

\begin{theorem}
We have:
\begin{align}
\Sigma_!(\pi) = \mu x. (id_{\sum_D \pi(D)} \vee\hspace{-0,5cm} \bigvee_{d:D_1\to D_2\in \mathbb{D}}\hspace{-0,5cm} x\circ \overline{\pi(d)}).\label{id:lax_functor_saturation}
\end{align}
\end{theorem}
\begin{proof}
At first observe that the assignment 
\begin{align*}
F:\mathsf{K}(X,X)\to \mathsf{K}(X,X); F(x)= id_{X} \vee \bigvee_{d\in \mathbb{D}} x\circ \overline{\pi(d)} 
\end{align*}
is well defined for any morphism $x\in \mathsf{K}(X,X)$. Indeed, $\bigvee_{d\in \mathbb{D}} x\circ \overline{\pi(d)}$ exists since it can be rewritten as $\bigvee x\circ \Pi = x\circ \bigvee\Pi.$
Moreover, its least fixed point is given by $\bigvee_{l\in \mathbb{N}} F^l(id)$. Additionally, we have:
$$
(\bigvee \Pi)^l = F^l(id) \text{ for any }l\in \mathbb{N}.
$$
The above assertion is true for $l=1$. Now, by induction, assume it holds for $l$ and consider:
\begin{align*}
&(\bigvee \Pi)^{l+1} =  (\bigvee \Pi)^{l}\circ \bigvee \Pi=  F^{l}(id)\circ (\bigvee \Pi)  =id \vee F^{l}(id)\circ (\bigvee \Pi)=\\
&id \vee \bigvee F^{l}(id)\circ  \Pi  = id \vee \bigvee \{ F^{l}(id)\circ (\overline{\pi(d_1)}\vee \ldots \vee \overline{\pi(d_k)})\mid k\in \{1,2,\ldots\}, d_i\in \mathbb{D}\}=\\
&id \vee \bigvee \{ F^{l}(id)\circ \overline{\pi(d_1)}\vee \ldots \vee F^{l}(id)\circ \overline{\pi(d_k)}\mid k\in \{1,2,\ldots\}, d_i\in \mathbb{D}\}=\\
& id \vee \bigvee_{d\in \mathbb{D}} F^{l}(id)\circ \overline{\pi(d)} = F^{l+1}(id).
\end{align*}
This completes the proof, as $\Sigma_!(\pi) = \bigvee_{l\in \mathbb{N}}(\bigvee \Pi)^l = \bigvee_l F^l(id)$.
\end{proof}

\begin{theorem} \label{theorem:omega_cpo_enrichment_adjunction}
If $\mathsf{K}$ is a left-distributive $\omega\mathsf{Cpo}^\vee$-enriched category then the functor $[!,\mathsf{K}]^J:[1,\mathsf{K}]^J\to [\mathbb{N},\mathsf{K}]^J$ admits a left 2-adjoint $\Sigma_!$.
\end{theorem}
\begin{proof} The proof of this theorem is very similar to the proof of Theorem \ref{theorem:dcpo_left_adjoint}. In this case, however, we have $\bigvee\Pi  =\bigvee_{n\in \mathbb{N}} \bigvee_{k=1}^n \pi_k = \bigvee_{n\in\mathbb{N}}\pi_n$ for $\pi\in [\mathbb{N},\mathsf{K}]$. Hence, the assumption of $\mathsf{DCpo}^\vee$-enrichment can be naturally replaced with $\omega\mathsf{Cpo}^\vee$-enrichment, as only the suprema of $\omega$-chains are considered and no cotupling is used.
\end{proof}

\section{Weak bisimulation}\label{section:weak_bisimulation_lax_functors}

The primary purpose of this section is to introduce the notion of weak bisimulation for lax functors. We believe that the lax functor weak bisimulation can serve as an extension of coalgebra weak bisimulation in future applications. In order to justify this statement we use an example of weak bisimulations of timed processes and Markov chain transition functors (see Example \ref{example:timed_systems_weak} and \ref{example:markov_weak_bisim} below for details).
In the second part of this section we revisit coalgebraic weak bisimulation from Subsection~\ref{subsection:weak_bis_coal} and argue that it is, in fact, a consequence of weak bisimulation for lax functors.

\subsection{Weak bisimulation for lax functors} \label{subsection:flow_weak_bisimulation} Here, we assume the following:
\begin{itemize}
\item  $\mathbb{D}$ is a small category,
\item  $\mathsf{K}$ is a small category\footnote{See Remark \ref{remark:smallness_problems} on smallness of $\mathsf{K}$.},
\item $J$ is a subcategory of $\mathsf{K}$ with all objects from $\mathsf{K}$,
\item $\mathsf{K}$ is $\mathsf{DCpo}^\vee$-enriched\footnote{whenever $\mathbb{D}=\mathbb{N}$ then all occurences of $\mathsf{DCpo}^\vee$ in this subsection can be replaced with $\omega\mathsf{Cpo}^\vee$ and all theorems remain true. See also Theorem \ref{theorem:dcpo_left_adjoint} and \ref{theorem:omega_cpo_enrichment_adjunction} for comparison.}, 
\item $\mathsf{K}$ admits arbitrary coproducts of families indexed by objects of $\mathbb{D}$ with $\mathbb{D}$-indexed cotupling preserving the order.
\end{itemize}
The first step to define weak bisimulation for members of $[\mathbb{D},\mathsf{K}]^J$ is to embed $\mathsf{K}$ into a category which yields saturation (cf. Subsection \ref{subsection:weak_bis_coal}).
\subsubsection{Locally reflective embedding of $\mathsf{K}$ into a left distributive category} \label{subsection:local_embedding} The main goal of this subsection is to describe a supercategory of $\mathsf{K}$ which is $\mathsf{DCpo}^\vee$-enriched and, additionally, left distributive. This construction  is entirely based on an idea presented in \cite[Sec. 3.1]{brengos2015:jlamp}. However, in \emph{loc. cit.} it is carried out in the context of a small $\omega\mathsf{Cpo}^\vee$-enriched category $\mathsf{K}$. Here, $\omega\mathsf{Cpo}^\vee$-enrichment is replaced with $\mathsf{DCpo}^\vee$-enrichment. Nevertheless, all properties of this category are proved in the same, straightforward, manner. 

Consider the category $\widetilde{\mathsf{K}}=[\mathsf{K},\mathsf{DCpo}^\vee]$ of lax functors $\mathsf{K}\to \mathsf{DCpo}^\vee$ and oplax transformations. For $\pi,\pi'\in [\mathsf{K},\mathsf{DCpo}^\vee]$ and two oplax transformations $f,g$ from $\pi$ to $\pi'$ define:
$$
f\leq g \iff f_X(x) \leq g_X(x) \text{ for any }X\in \mathsf{K} \text{ and }x\in \pi X.
$$
Since the order on hom-sets of $\widetilde{\mathsf{K}}$ is imposed by the component-wise order from $\mathsf{K}$ binary joins are given by $(f\vee g)_X:\pi X\to \pi' X; x\mapsto f_X(x) \vee g_X(x)$. It is easy to check that $f\vee g$ is an oplax transformation.
Similarly, the suprema of directed families of oplax transformations are component-wise suprema. A straightforward verification proves that such suprema are oplax transformations. Since $\mathsf{K}$ is $\mathsf{DCpo}$-enriched, suprema of directed families are preserved by the composition in $\widetilde{\mathsf{K}}$. Hence, $\widetilde{\mathsf{K}}$ is $\mathsf{DCpo}^\vee$-enriched. Moreover, we have the following theorem.
\begin{theorem}\label{theorem:left_dist_hat_K} 
The category $\widetilde{\mathsf{K}}$ is right distributive $\mathsf{DCpo}^\vee$-enriched. As a consequence, $\widetilde{\mathsf{K}}^{op}$ is left distributive $\mathsf{DCpo}^\vee$-enriched.
\end{theorem}
\begin{proof}
Right distributivity of $\widetilde{\mathsf{K}}$ follows from the fact that the order on and the composition of oplax transformations is
defined pointwise.
\end{proof}

For an object $X\in \mathsf{K}$ and a morphism $f:X\to X'\in \mathsf{K}$ define
$\widehat{X} = \mathsf{K}(X,-)$ and $\widehat{f}=\mathsf{K}(f,-)$. Explicitly, the functor $\widehat{X}:\mathsf{K}~\to~\mathsf{DCpo}^\vee$ maps any $Y\in \mathsf{K}$ to $\mathsf{K}(X,Y)$ and any $g:Y\to Y'$ is mapped onto $\widehat{X}(g):K(X,Y)\to K(X,Y'); h\mapsto g\circ h$. Moreover, $\widehat{f}:\widehat{X}'\to \widehat{X}$ is the natural transformation whose $Y$-component is:
$$
\widehat{f}_Y:\mathsf{K}(X',Y)\to \mathsf{K}(X,Y); h\mapsto h\circ f. 
$$

Put $\widehat{\mathsf{K}}$ to be the full subcategory of $\widetilde{\mathsf{K}}^{op}$ consisting only of objects of the form $\widehat{X}$ for some $X\in \mathsf{K}$. The assignment $\widehat{(-)}$ is a locally monotonic functor $\widehat{(-)}:\mathsf{K}\to \widehat{\mathsf{K}}$. For any two objects $X,Y\in \mathsf{K}$ define $\Theta:~\widehat{\mathsf{K}}(\widehat{X},\widehat{Y})\to \mathsf{K}(X,Y)$ by
$
\Theta(f) = f_Y(id_Y)$  for $f\in \widehat{\mathsf{K}}(\widehat{X},\widehat{Y})=\widetilde{\mathsf{K}}(\widehat{Y},\widehat{X})$.

\begin{theorem}
The assignment $\Theta$ is the left adjoint to the hom-object restriction $\widehat{(-)}:\mathsf{K}(X,Y)\to \widehat{\mathsf{K}}(\widehat{X},\widehat{Y})$ of $\widehat{(-)}:\mathsf{K}\to \widehat{\mathsf{K}}$. Hence, the functor $\widehat{(-)}:\mathsf{K}\to \widehat{\mathsf{K}}$ is a locally reflective embedding.
\end{theorem}
\begin{proof}
It is straightforward to show that for any $g\in \mathsf{K}(X,Y)$ and $f \in \widehat{\mathsf{K}} (\widehat{X},\widehat{Y})$ we have:
$\Theta_{X,Y}(f) =  f_Y(id_Y) \leq g \iff \phi \leq \widehat{g}$. This ends the proof.
\end{proof}

\begin{theorem}\label{theorem:coproducts_in_hat_K}
The category $\widehat{\mathsf{K}}$ admits arbitrary coproducts of families indexed by objects of $\mathbb{D}$ with cotupling preserving the order.
\end{theorem}
\begin{proof}
The coproducts in $\widehat{\mathsf{K}}$ come from $\mathsf{K}$. Indeed, let $\{\widehat{X}_D\}_{D\in \mathbb{D}}$ be a family of objects in $\widehat{\mathsf{K}}$. We will now show that its coproduct $\sum_D \widehat{X}_D$ in $\widehat{\mathsf{K}}$ is given by $\widehat{\sum_D X_D}$. At first, observe that for $\mathsf{in}_{D}:X_D\to \sum_D X_D$ the transformation $\widehat{\mathsf{in}_D}$ is an oplax transformation from $\widehat{\sum_D X_D}$ to $\widehat{X}_D$. Secondly, consider $\widehat{X}$ with a family morphisms $\psi_D\in \widehat{\mathsf{K}}(\widehat{X}_D, \widehat{X})$. By the definition of $\mathsf{\widehat{K}}$ these morphisms are oplax transformations $\psi_D:\widehat{X}\to \widehat{X}_D$. The $Y$-components $(\psi_D)_Y:\mathsf{K}(X,Y)\to\mathsf{K}(X_D,Y)$ of $\psi_D$ satisfy for any $h:X\to Y$ and $g:Y\to Y'$
$$
(\psi_D)_Y(g\circ h) \leq g\circ (\psi_D)_Y(h).
$$
Consider a family of maps $\psi_Y:\mathsf{K}(X,Y)\to \mathsf{K}(\sum_D X_D,Y)$ indexed by $Y\in \mathsf{K}$ which is defined for any $h:X\to Y$ by:
$$\psi_Y(h) = [(\psi_{D})_Y(h)].$$  
The family $\psi = \{\psi_Y\}_{Y\in \mathsf{K}}$ is an oplax transformation from $\widehat{X}$ to $\widehat{\sum_D X_D}$ since cotupling in $\mathsf{K}$ preserves the order. Moreover, it is a unique oplax transformation making $\psi_D=\widehat{\mathsf{in}_D}\circ \psi$. This precisely means that $\widehat{\sum_D X_D}$ is the product in the full subcategory of $\widetilde{\mathsf{K}}=[\mathsf{K},\mathsf{DCpo}^\vee]$ whose objects are functors of the form $\widehat{X}$. Hence, it is a coproduct in $\widehat{\mathsf{K}}$. The cotupling in $\widehat{\mathsf{K}}$ preserves the order as the order in $\widehat{\mathsf{K}}$ is inherited from $\mathsf{K}$.
\end{proof}

We are now ready to summarize this paragraph. By Theorem \ref{theorem:left_dist_hat_K}, \ref{theorem:coproducts_in_hat_K} and \ref{theorem:dcpo_left_adjoint} the functor $[!,\widehat{\mathsf{K}}]:[1,\widehat{\mathsf{K}}]\to [\mathbb{D},\widehat{\mathsf{K}}]$ admits the left $2$-adjoint $\Sigma_!$:
\begin{align}
\xymatrix{
 [\mathbb{D},\widehat{\mathsf{K}}] \ar@/^1pc/[r]^{\Sigma_!} \ar@{}[r]|\perp & [1,\widehat{\mathsf{K}}]\ar@/^1pc/[l]^{[!,\widehat{\mathsf{K}}]}.
} \label{adjunction_weak_bis_lax}
\end{align}
This observation allows us to introduce the notion of weak bisimulation for lax functors in $[\mathbb{D},\mathsf{K}]^J$.

\subsubsection{Weak behavioural morphisms and weak bisimulation} \label{subsubsection:weak_bisim:lax} We will now define weak bisimulation on lax functors in $[\mathbb{D},\mathsf{K}]^J$. As in Subsection \ref{subsection:weak_bis_coal}, weak bisimulation will be defined as a kernel pair of a weak behavioural morphism. Hence, we start with the definition of the latter. Let  $\pi\in [\mathbb{D},\mathsf{K}]^J$ be a lax functor. Put $\widehat{\pi}=\widehat{(-)}\circ \pi:\mathbb{D}\to \widehat{\mathsf{K}}$ and $X=\sum_{D\in \mathbb{D}} \pi(D)$, where the coproduct is calculated in $\mathsf{K}$. Note that weak behavioural morphisms and weak bisimulation on $\pi$ considered below are defined on the carrier of $\Sigma_!(\pi)$, i.e. on the object $X$.
\begin{definition}
 We say that an arrow $f:X\to Y$ in $J$ is \emph{weak behavioural morphism} on $\pi$  provided that there is an endomorphism  $\beta:Y\to Y \in \mathsf{K}$ such that:
\begin{align}
\Theta(\widehat{f}\circ \Sigma_!(\widehat{\pi})) = \Theta(\widehat{\beta}\circ \widehat{f}). \label{equation:wb_lax}
\end{align}
A relation $R\rightrightarrows X$ is called \emph{weak bisimulation} on $\pi$ provided that it is a kernel pair of a weak behavioural morphism on $\pi$.
\end{definition}

The following results will lead us to a  simplification of the equation (\ref{equation:wb_lax}). 
\begin{lemma}\label{lemma:identity_theta}
$$
\Theta(\widehat{f}\circ \Sigma_!(\widehat{\pi})) = \mu x. (f\vee \hspace{-0.5cm} \bigvee_{d:D_1\to D_2\in \mathbb{D}}\hspace{-0.5cm}x\circ \overline{\pi(d)}).
$$
\end{lemma}
\begin{proof}
At first observe that the assignment $F:\mathsf{K}(X,Y)\to \mathsf{K}(X,Y)$ given by $F(x)=f\vee  \bigvee_{d:D_1\to D_2\in \mathbb{D}}x\circ \overline{\pi(d)}$ for $x\in \mathsf{K}(X,Y)$ is well defined. Indeed, the supremum $\bigvee_{d:D_1\to D_2\in \mathbb{D}}x\circ \overline{\pi(d)}$ exists since it can be rewritten as: $$\bigvee\{x\circ \overline{\pi(d_1)}\vee\ldots \vee x\circ \overline{\pi(d_k)}\mid k\in \{1,2,\ldots\}, d_i\in \mathbb{D}\}.$$
It is easy to see that the least fixed point of $F$ is $\bigvee_{n\in \mathbb{N}} F^n(f)$. Let $$\widehat{\Pi} = \{\overline{\widehat{\pi}(d_1)} \vee \ldots \vee \overline{\widehat{\pi}(d_k)}\mid k\in \{1,2,\ldots\}, d_i \text{ is a morphism in } \mathbb{D}\}.$$ Since $\bigvee\widehat{\Pi} = \bigvee_{d\in \mathbb{D}}\overline{ \widehat{\pi}(d)}$ we have:
\begin{align*}
&\Theta(\widehat{f}\circ \Sigma_!(\widehat{\pi})) = \Theta( \widehat{f} \circ \bigvee_{l\in \mathbb{N}}(\bigvee\widehat{\Pi})^l)=\Theta ( \bigvee_{l\in \mathbb{N} }\widehat{f} \circ (\bigvee_{d\in \mathbb{D}}\overline{ \widehat{\pi}(d)} )^l) \stackrel{(i)}{=}\\
&\bigvee_{l\in \mathbb{N}} \Theta(\widehat{f} \circ (\bigvee_{d\in \mathbb{D}}\overline{ \widehat{\pi}(d)} )^l)    \stackrel{(ii)}{=} 
\bigvee_{l\in \mathbb{N}} F^l(f) = \mu x. (f\vee \hspace{-0.5cm}\bigvee_{d:D_1\to D_2\in \mathbb{D}}\hspace{-0.5cm}x\circ \overline{\pi(d)}).
\end{align*}
The equation  $(i)$ follows by the fact that $\Theta$ preserves arbitrary suprema (as it is a left adjoint). The identity $(ii)$ follow by induction. For $l=0$ it is vacuously true. Assume that $\Theta(\widehat{f} \circ (\bigvee_{d\in \mathbb{D}}\overline{ \widehat{\pi}(d)} )^l)  = F^l(f)$ for a natural number $l$ and consider:
\begin{align*}
&\Theta(\widehat{f} \circ (\bigvee_{d\in \mathbb{D}}\overline{ \widehat{\pi}(d)} )^{l+1}) = \Theta(\widehat{f} \circ (\bigvee_{d\in \mathbb{D}}\overline{ \widehat{\pi}(d)} )^{l}\circ (\bigvee_{d\in \mathbb{D}}\overline{ \widehat{\pi}(d)} )) \stackrel{(a)}{=} \\
&\Theta(\bigvee_{d'\in \mathbb{D}} \widehat{f} \circ (\bigvee_{d\in \mathbb{D}} \overline{ \widehat{\pi}(d)} ))^{l}\circ \overline{ \widehat{\pi}(d')} )) \stackrel{(b)}{=}\bigvee_{d'\in \mathbb{D}} \Theta( \widehat{f} \circ (\bigvee_{d\in \mathbb{D}} \overline{ \widehat{\pi}(d)} ))^{l}\circ \overline{ \widehat{\pi}(d')} ))  \stackrel{(c)}{=}\\
& =\bigvee_{d'\in \mathbb{D}} F^l(f) \circ \overline{\pi(d')}\stackrel{(d)}{=}f\vee \bigvee_{d'\in \mathbb{D}} F^l(f) \circ \overline{\pi(d')}=F^{l+1}(f).
\end{align*}
The identity $(a)$ follows by left distributivity and $\mathsf{DCpo}^\vee$-enrichment of $\widehat{\mathsf{K}}$. The equation $(b)$ is a consequence of $\Theta$ being a left adjoint. The equality $(c)$ follows by $\overline{\widehat{\pi}(d)} = \widehat{\overline{\pi(d)}}$. Finally, the identity $(d)$ is a consequence of $F^l(f)\leq F^{l+1}(f)$, $F^0(f) = f$ and $id_X\leq \overline{\pi(id_D)}$ for any $D\in \mathbb{D}$:
$$
f\leq F^l(f) \leq F^l(f)\circ \overline{\pi(id_D)} \leq \bigvee_{d\in \mathbb{D}} F^l(f)\circ \overline{\pi(d)}.
$$
\end{proof}
Finally, since $\Theta(\widehat{\beta}\circ \widehat{f})=\beta\circ f$, by Lemma \ref{lemma:identity_theta} the equation (\ref{equation:wb_lax}) becomes:
\begin{align}
\mu x. (f \vee \hspace{-0.5cm} \bigvee_{d:D_1\to D_2\in \mathbb{D}}\hspace{-0.5cm}x\circ \overline{\pi(d)}) = \beta\circ f.\label{equation:lfps}
\end{align}
If we additionally assume $\mathsf{K}$ satisfies left distributivity then the equation (\ref{equation:lfps}) can be simplified even further. In this case we have the following. 
\begin{theorem} \label{theorem:weak_bis_left_distributive} If $\mathsf{K}$ is left distributive then (\ref{equation:lfps}) becomes:
\begin{align}
f\circ \mu x. (id_X \vee \hspace{-0.5cm} \bigvee_{d:D_1\to D_2\in \mathbb{D}}\hspace{-0.5cm}x\circ \overline{\pi(d)}) = \beta \circ f.\label{equation:left_dist_ps}
\end{align}
\end{theorem}
\begin{proof}
Consider assignments $F:\mathsf{K}(X,Y)\to \mathsf{K}(X,Y);x\mapsto f\vee \bigvee_{d:D_1\to D_2\in \mathbb{D}} x\circ \overline{\pi(d)})$ and $G:\mathsf{K}(X,X)\to \mathsf{K}(X,X);x\mapsto id_X \vee  \bigvee_{d:D_1\to D_2\in \mathbb{D}}x\circ \overline{\pi(d)}$. The left hand side of (\ref{equation:lfps}) is given by $\bigvee_{n\in \mathbb{N}} F^n(f)$ and the left hand side of (\ref{equation:left_dist_ps}) is $f\circ \bigvee_{n\in \mathbb{N}}G^n(id_X) = \bigvee_{n\in \mathbb{N}}f\circ G^n(id_X)$. We will now inductively show that for any $n\in \mathbb{N}$:
$$
F^n(f) = f\circ G^n(id_X).
$$
The assertion is true for $n=0$. Assume it holds for $n$ and consider $F^{n+1}(f) = f \vee \bigvee_{d\in \mathbb{D}}F^n(f)\circ \overline{\pi(d)}= f \vee \bigvee_{d\in \mathbb{D}}f\circ G^n(id_X)\circ \overline{\pi(d)}\stackrel{(i)}{=} f \vee f\circ \bigvee_{d\in \mathbb{D}} G^n(id_X)\circ \overline{\pi(d)} = f\circ(id_X\vee \bigvee_{d\in \mathbb{D}} G^n(id_X)\circ \overline{\pi(d)}) = f\circ G^{n+1}(id_X)$. The identity $(i)$ follows by left distributivity and $\mathsf{DCpo}^\vee$-enrichment of $\mathsf{K}$.
\end{proof}

\begin{remark}\label{remark:saturation_left_distrib}
Here, we want to discuss Theorem \ref{theorem:weak_bis_left_distributive} and its interpretation. If $\mathsf{K}$ is additionally left distributive then by Theorem~\ref{theorem:dcpo_left_adjoint} the change-of-base functor $[!,\mathsf{K}]$ admits a left adjoint $\Sigma_!$ with $\Sigma_!(\pi) = \mu x.(id_X\vee \bigvee_{d\in \mathbb{D}} x\circ \overline{\pi(d)})$ for any $\pi\in [\mathbb{D},\mathsf{K}]$. Therefore, in the light of the above, Theorem \ref{theorem:weak_bis_left_distributive} states that whenever $\mathsf{K}$ is left distributive weak bisimulation saturation can be performed on the level of $\mathsf{K}$.  In this case, weak behavioural morphisms on $\pi$ are simply strong homomorphisms whose domain is the endomorphism $\Sigma_!(\pi):X\to X$.
\end{remark}
\subsection{Weak bisimulation on $M$-flows}
Whenever $\mathbb{D}=M$ is a one-object category induced by a monoid $M=(M,\cdot,1)$ then the last assumption from Subsection~\ref{subsection:flow_weak_bisimulation} is vacously true. Moreover, for any weak behavioural morphism on a lax functor $\pi\in [M,\mathsf{K}]^J$ its domain matches the carrier of $\pi$. Hence,  $\overline{\pi_m} = \pi_m$ for any $m\in M$ and (\ref{equation:lfps}) becomes:

\begin{align}
\mu x. (f \vee  \bigvee_{m\in M}x\circ \pi_m) = \beta\circ f.\label{equation:lfps_flows}
\end{align}

\subsubsection{Examples of coalgebra flow weak bisimulations} From the point of view of coalgebra, the most prominent application of the theory presented above is weak bisimulation for elements of the category of $\mathbb{N}$-flows $[\mathbb{N},\mathcal{K}l(T)]^\mathsf{C}$. We devote a separate subsection to it (see Subsection \ref{subsec:weak_lax_coalgebra}). Below we present two examples of weak bisimulation for coalgebra $M$-flows with $M\neq \mathbb{N}$. In the first example we show that weak bisimulation for  $\underline{\alpha}$ from Example \ref{example:semantics_of_timed} viewed as a member of $[\mathbb{N}\times [0,\infty),\mathcal{K}l(\mathcal{P}(\Sigma_\tau\times \mathcal{I}d))]^\mathsf{Set}$ coincides with time-abstract bisimulation for timed processes \cite{Larsen199775}. In the second example we instantiate the definition of weak bisimulation on transition functors of homogeneous continuous Markov chains from Example \ref{example:markov_chain_transition}.

\begin{example}[Weak bisimulation(s) for timed processes]\label{example:timed_systems_weak}
Let us now go back to the semantics of timed processes from Example \ref{example:semantics_of_timed}. We will introduce two definitions of bisimulation for the semantics LTS $X\to \mathcal{P}([\Sigma_\tau \cup (0,\infty)]\times X)$ \cite{Larsen199775} and show that they coincide with our, lax functorial, weak bisimulations. We need to introduce the following notation first:
\begin{enumerate}
\item $x\stackrel{\tau}{\implies} y$ if $x(\stackrel{\tau}{\to})^\ast y$,
\item $x\stackrel{a}{\implies} y$ if $x\stackrel{\tau}{\implies} \circ \stackrel{a}{\to} \circ \stackrel{\tau}{\implies} y$ for $a\in \Sigma$,
\item $x\stackrel{t}{\implies} y$ if $x\stackrel{\tau}{\implies} \circ \stackrel{t_1}{\to} \circ \stackrel{\tau}{\implies} \circ \ldots \stackrel{\tau}{\implies} \circ \stackrel{t_n}{\to} \circ \stackrel{\tau}{\implies} y$ for $t,t_i\in (0,\infty)$ and $t=t_1+\ldots+t_n$,
\item $x\stackrel{\tau}{\multimap} y$ if $x\stackrel{t}{\implies} y$ for some $t\in (0,\infty)$,
\item $x\stackrel{a}{\multimap} y$ if $x\stackrel{\tau}{\multimap} \circ \stackrel{a}{\to} \circ \stackrel{\tau}{\multimap} y$ for $a\in \Sigma$.
\end{enumerate}
An equivalence relation $R$ on the set of timed processes $X$ is called \emph{weak timed bisimulation} \cite{Larsen199775} provided that whenever $(x,x')\in R$ then for all $\sigma \in \Sigma_\tau\cup (0,\infty)$:
$$
x\stackrel{\sigma}{\implies} y \text{ implies } \exists y' \text{ s.t. } x'\stackrel{\sigma}{\implies} y' \text{ and }(y,y')\in R. 
$$ 
The relation $R$ is \emph{weak time-abstract bisimulation} \cite{Larsen199775} provided that whenever $(x,x')\in R$ then for all $\sigma\in \Sigma_\tau$:
$$
x\stackrel{\sigma}{\multimap} y \text{ implies } \exists y' \text{ s.t. } x'\stackrel{\sigma}{\multimap} y' \text{ and }(y,y')\in R. 
$$ 
\begin{theorem}
An equivalence relation $R$ on the set of timed processes $X$   is a weak bisimulation on the lax functor $\underline{\alpha}\in [\mathbb{N},\mathcal{K}l(\mathcal{P}^{\Sigma,[0,\infty)})]$ if and only if it is a weak timed bisimulation.  The relation $R$ is a weak bisimulation on $\underline{\alpha}$ viewed as a member of $[\mathbb{N}\times [0,\infty),\mathcal{K}l(\mathcal{P}(\Sigma_\tau\times \mathcal{I}d))]^\mathsf{Set}$ if and only if it is a weak time-abstract bisimulation.
\end{theorem}
\begin{proof} We only sketch the proof of the second statement. The first follows in an analogous manner. In the light of Theorem \ref{theorem:weak_bis_left_distributive} and Remark \ref{remark:saturation_left_distrib} it is enough to show that $R$ is a weak time-abstract bisimulation if and only if $R$ is a strong bisimulation on $\alpha^T:X\to \mathcal{P}(\Sigma_\tau\times X)$ as the functor $\Sigma_!$ from:
$$
\xymatrix{
 [\mathbb{N}\times [0,\infty),\mathcal{K}l(\mathcal{P}^{\Sigma,1})] \ar@/^1pc/[r]^{\Sigma_!} \ar@{}[r]|\perp &   [1,\mathcal{K}l(\mathcal{P}^{\Sigma,1})]\ar@/^1pc/[l]^{(-)\circ !}  
}
$$
maps $\underline{\alpha}$ to $\alpha^T$.
Assume $R$ is a weak time-abstract bisimulation on $X$. Take $(x,y)\in R$ and consider $x\stackrel{\sigma}{\to}_{\alpha^T} x'$ for $\sigma\in \Sigma_\tau$. This means that $x\stackrel{(\sigma,t)}{\to}_{\alpha^\ast} y$ which implies that either  $x\stackrel{(\tau,t_1)}{\to}_{\alpha} x_1 \stackrel{(\tau,t_2)}{\to}_{\alpha} \ldots \stackrel{(\tau,t_n)}{\to}_{\alpha} x_n  = x'$ with $t=t_1+\ldots + t_{n}$ and $x_i\in X$ for $\sigma=\tau$ or
$x\stackrel{(\tau,t_1)}{\to}_{\alpha^\ast} x_1 \stackrel{(\sigma,t_2)}{\to}_{\alpha} x_2 \stackrel{(\tau,t_3)}{\to}_{\alpha^\ast} x'$ for some $x_1,x_2\in X$, $\sigma \in \Sigma$ and $t=t_1+t_2+t_3$. In both cases this implies $x\stackrel{\sigma}{\multimap}x'$. Since $R$ is a time-abstract bisimulation there is $y'$ such that $y\stackrel{\sigma}{\multimap} y'$ and $(y,y')\in R$. But this also means that $y\stackrel{\sigma}{\to}_{\alpha^T} y'$. Hence, $R$ is a strong bisimulation on $\alpha^T$. The implication in the opposite direction is proved similarly. 
\end{proof}
\end{example}

\begin{example}[Weak bisimulation for CTMC's transition functors]\label{example:markov_weak_bisim} Here, we continue Example~\ref{example:markov_chain_transition} and characterize weak bisimulation on the transition functor $\pi= (\pi_t)_{t\geq 0}$ of the homogeneous CTMC $(X_t)_{t\geq 0}$.  Consider an equivalence relation $R$ on the state space $S$. For an abstract class $C$ of $R$ let us denote:
$$
p_{i,C}^t=\mathbb{P}(X_r\in C \text{ for some } r\geq t\mid X_0=i) \text{ and } p_{i,C} = p_{i,C}^0.
$$
\begin{lemma}\label{lemma:homogen_chain}
If $(X_t)_{t\geq 0}$ is homogeneous then $\{p_{i,C}\}_{i\in S}$ satisfies:
\begin{align}
p_{i,C}= \left \{\begin{array}{cc} 1 & \text{ if } i\in C,\\ \sup_{t\geq 0} \sum_{j\in S} p_{j,C}\cdot p_{i,j}(t) & \text{ otherwise.}  \end{array}\right.
\end{align}
\end{lemma}
\begin{proof}
It is clear that if $i\in C$ then $p_{i,C}=1$. For $i\notin C$ we have:
\begin{align*}
&p_{i,C} = \sup_{t\geq 0} p_{i,C}^t = \sup_{t\geq 0} \sum_{j\in S}\mathbb{P}(X_r\in C, r\geq t\mid X_t=j,X_0=i) \cdot \mathbb{P}(X_t=j\mid X_0=i) \stackrel{\dagger}{=}\\
&\sup_{t\geq 0} \sum_{j\in S}\mathbb{P}(X_r\in C, r\geq t\mid X_t=j) \cdot \mathbb{P}(X_t=j\mid X_0=i) \stackrel{\dagger\dagger}{=}\\
&\sup_{t\geq 0} \sum_{j\in S}\mathbb{P}(X_r\in C, r\geq 0\mid X_0=j) \cdot \mathbb{P}(X_t=j\mid X_0=i) =\sup_{t\geq 0} \sum_{j\in S}p_{j,C}\cdot p_{i,j}(t).
\end{align*}
The identity $(\dagger)$ follows by $(X_t)_{t\geq 0}$ being markovian and $(\dagger\dagger)$ by homogeneity of the given process.
\end{proof}
\begin{theorem}
 The relation $R$ is a weak bisimulation on the transition functor $\pi$  of a homogeneous CTMC $(X_t)_{t\geq 0}$ provided that for any $(i,j)\in R$ and any  abstract class $C$ of $R$ we have:
\begin{align}
\mathbb{P}(X_t\in C \text{ for some } t\geq 0\mid X_0=i) = \mathbb{P}(X_t\in C \text{ for some } t\geq 0\mid X_0=j).\label{id:equal_probab}
\end{align}
\end{theorem}
\begin{proof}
Let $f:S\to S_{/R}; i\mapsto [i]_{/R}$. By Lemma \ref{lemma:homogen_chain}:
$$
\mu x. (f^\sharp \vee \bigvee_{r\in [0,\infty)} x\circ \pi_r):S\to \mathbb{F}_{[0,\infty]}(S_{/R}),i\mapsto \sum_{C\in S_{/R}}p_{i,C}\cdot C.
$$
In other words, $\mu x. (f^\sharp \vee \bigvee_{r\in [0,\infty)} x\circ \pi_r)(i)(C) = p_{i,C}$. Satisfaction of the identity (\ref{id:equal_probab}) is equivalent to existence of an $\mathbb{F}_{[0,\infty]}$-coalgebra $\beta:S_{/R}\to \mathbb{F}_{[0,\infty]}S_{/R}$ which makes $\mu x. (f^\sharp \vee \bigvee_{r\in [0,\infty)} x\circ \pi_r) = \beta \circ f^\sharp$ hold.
\end{proof}

\end{example}

\subsubsection{Cumulative behaviour between members of different flow categories} Although in this paper we consider only $M$-flows as examples of lax functors, weak bisimulation from Subsection \ref{subsubsection:weak_bisim:lax} is the defined in a more general setting (i.e. for lax functors whose domain is arbitrary small category $\mathbb{D}$). This level of generality can be easily justified. It is interesting to note that, in particular, the setting allows us to compare cumulative behaviour between lax functors on not necessarily the same domains. To see this consider two lax functors $\pi_1\in [\mathbb{N},\mathsf{K}]$ and $\pi_2\in [[0,\infty),\mathsf{K}]$ whose carriers are $X$ and $Y$ respectively. They naturally impose a lax functor $\pi$ from the category $\mathbb{N}+[0,\infty)$ on two objects to $\mathsf{K}$. Weak bisimulation on $\pi$ is a relation on $X+Y$ which compares cumulative behaviour of $\pi_1$ and $\pi_2$ simultaneously.

\subsection{Weak bisimulation for coalgebras revisited}\label{subsec:weak_lax_coalgebra} This subsection is devoted to the connection between coalgebraic weak bisimulation from Subsection \ref{subsection:weak_bis_coal} and weak bisimulation on members of $[\mathbb{N},\mathcal{K}l(T)]^\mathsf{C}$ or, in general, of $[\mathbb{N},\mathsf{K}]^J$.

Coalgebraic saturation described in Subsection \ref{subsection:weak_bis_coal} is given by the adjunction (\ref{adjunction_coalgebraic_saturation}). It is interesting to note that it can be considered a consequence of an adjunction between suitable lax functor categories. Indeed, if $\mathsf{K}$ is a small $\omega\mathsf{Cpo}^\vee$-enriched category then by Proposition \ref{proposition_2_T_omega} we have:
$$
\xymatrix{
\mathsf{End}^\leq (\widehat{\mathsf{K}})\ar@/^1pc/[r]^{\underline{(-)}} \ar@{}[r]|\perp & [\mathbb{N},\widehat{\mathsf{K}}] \ar@/^1pc/[l]^{{(-)}_1}\ar@/^1pc/[r]^{\Sigma_!} \ar@{}[r]|\perp & [1,\widehat{\mathsf{K}}]\ar@/^1pc/[l]^{[!,\widehat{\mathsf{K}}]}
}.
$$
Since $[1,\widehat{\mathsf{K}}]\cong\mathsf{End}^{\leq\ast} (\widehat{\mathsf{K}})$ the composition of the above adjunctions yields (\ref{adjunction_coalgebraic_saturation}).

\begin{theorem}
Let $\alpha:X\to X$ be an endomorphism in $\mathsf{K}$. A relation $R\rightrightarrows X$ is a weak bisimulation on $\alpha$ if and only if $R$ is a weak bisimulation on $\underline{\alpha}:\mathbb{N}\to \mathsf{K}$. 
\end{theorem}
\begin{proof}
This follows directly by the definition of weak bisimulation for endomorphisms and lax functors and the fact that
$$
\alpha^\ast_f =\mu x . (f \vee x\circ \alpha)= \mu x . (f \vee \bigvee_{n\in \mathbb{N}} x\circ \alpha^n).
$$
\end{proof}

\section{Summary}\label{section:summary}
We presented the framework of lax functors as a setting that generalizes the setting of endomorphisms in which we can introduce the notion of weak bisimulation. Just like a single endomorphism (understood here as a coalgebra with silent moves) is a process with discrete time, a lax functor can
be viewed as e.g. continuous time process or a system of processes. We showed that in many cases, the change-of-base functor between lax functor categories admits a left adjoint and that the
adjunction $[\mathbb{N},\mathsf{K}]\leftrightarrows [1,{\mathsf{K}}]$ plays an important role in coalgebraic saturation and  weak bisimulation. Using the adjunction (\ref{adjunction_weak_bis_lax}) we introduced the
notion of weak bisimulation on a lax functor. This relation
takes into account its cumulative behaviour. 

We plan to investigate to what extent the setting of lax functors is applicable. Indeed, it seems there is plethora of examples of timed structures found in the literature ranging from stochastic timed automata semantics \cite{BBBMBGJ-lmcs14} to generalized flow systems \cite{davoren07} that could possibly fit it. 
\subsubsection*{Acknowledgements}
I express my gratitude to Marco Peressotti for fruitful discussions on timed processes. I am very grateful to Agnieszka Piliszek for helping me sort out the Markov chain example. I would like to thank Tony Barrett for his linguistic support. Finally, I truly appreciate the anonymous referees for their valuable remarks and comments.

\bibliographystyle{abbrv}
\bibliography{biblio}

\end{document}